\documentclass[ba]{imsart}

\RequirePackage{amsthm,amsmath,amsfonts,amssymb}

\usepackage{float}

\usepackage{xcolor}
\usepackage{bm}
\usepackage[ruled,vlined,linesnumbered]{algorithm2e}

\usepackage{subcaption}
\usepackage{booktabs}             
\usepackage{multirow}             
\usepackage{colortbl}  
\usepackage{comment}

\newcommand{\bigTheta}{\boldsymbol{\Theta}}

\usepackage{enumitem}

\RequirePackage[authoryear]{natbib}
\RequirePackage[colorlinks,citecolor=blue,urlcolor=blue,backref=page,backref=page]{hyperref}

\usepackage{hyperref}
\usepackage{url}

\RequirePackage{graphicx}

\usepackage{xr-hyper}

\pubyear{}
\arxiv{}
\volume{}
\issue{}
\firstpage{}
\lastpage{}

\startlocaldefs
\theoremstyle{plain}
\newtheorem{theorem}{Theorem}[section]

\newtheorem{lemma}[theorem]{Lemma}

\theoremstyle{definition}

\theoremstyle{remark}

\endlocaldefs

\begin{document}

\begin{frontmatter}
\title{Multiple Jump MCMC: A Scalable Algorithm for Bayesian Inference on Binary Model Spaces}
\runtitle{Scalable MCMC algorithm for Bayesian inference}

\begin{aug}
\author[A]{\fnms{Lucas}~\snm{Vogels}\ead[label=e1]{l.f.o.vogels@uva.nl}\orcid{0009-0005-3715-3937}},
\author[A]{\fnms{Reza}~\snm{Mohammadi}\ead[label=e2]{a.mohammadi@uva.nl}\orcid{0000-0001-9538-0648}},
\author[A]{\fnms{Marit}~\snm{Schoonhoven}\ead[label=e3]{}\orcid{0000-0002-2268-7486}}
\author[B]{\fnms{Sinan}~\snm{Yıldırım}\ead[label=e4]{}\orcid{0000-0001-7980-8990}}
\and
\author[A]{\fnms{Ş. İlker}~\snm{Birbil}\ead[label=e5]{}\orcid{0000-0001-7472-7032}}

\address[A]{Amsterdam Business School, University of Amsterdam\printead[presep={,\ }]{e1,e2}}
\address[B]{Industrial Engineering, Sabanci University \printead[presep={\ }]{e4}}
\runauthor{L. Vogels et al.}
\end{aug}

\begin{abstract}
This article considers Bayesian model inference on binary model spaces. Binary model spaces are used by a large class of models, including graphical models, variable selection, mixture distributions, and decision trees. Traditional strategies in this field, such as reversible jump or birth-death MCMC algorithms, are still popular, despite suffering from a slow exploration of the model space. In this article, we propose an alternative: the Multiple Jump MCMC algorithm. The algorithm is simple, rejection-free, and remarkably fast. When applied to undirected Gaussian graphical models, it is $100$ to $200$ times faster than the state-of-the-art, solving models with $500,000$ parameters in less than a minute. We provide theorems showing how accurately our algorithm targets the posterior, and we show how to apply our framework to Gaussian graphical models, Ising models, and variable selection, but note that it applies to most Bayesian posterior inference on binary model spaces. 
\end{abstract}



\end{frontmatter}

\section{Introduction}

This article introduces a new MCMC algorithm for Bayesian model inference on binary model spaces. That is, we consider a set of candidate models, each with an associated parameter vector, and quantify the posterior for each model. We assume that each candidate model can be represented by a binary vector. This set-up is used by a large class of Bayesian model inference problems, including variable selection \citep{Tadesse2021}, mixture distributions \citep{Stephens2000}, regression trees \citep{Mohammadi2020}, and graphical models. Examples of graphical models include undirected Gaussian graphical models \citep{vogels2024}, Ising models \citep{Marsman2022}, or mixed graphical models \citep{Wang2023}. 

On binary model spaces (and on general model spaces), Markov chain Monte Carlo (MCMC) methods are a popular tool for Bayesian model inference. These methods often circumvent the computation of normalizing constants and are suitable for exploring larger model spaces. MCMC methods design a Markov chain that converges to the posterior distribution. When the Markov chain is run long enough, the states in the tail of the Markov chain correspond to samples from the posterior.  

MCMC methods generally focus on uncovering the joint posterior, i.e. the posterior distribution of both the model and the associated parameters. When the associated parameter vectors have different dimensions, these MCMC methods are also referred to as trans-dimensional. They exist in discrete-time \citep{Green1995} or continuous-time \citep{Mohammadi2015}. Alternatively, one can integrate out the parameter vector and focus instead on the marginal posterior. This strategy creates an MCMC solely over the model space, see \citet{Madigan1995} for a discrete-time example and \cite{Mohammadi2024} for a continuous-time example. 

Both trans-dimensional and marginal methods are either reversible or non-reversible. Reversible methods depend on a condition called detailed balance, which states that every move must balance with its reverse move. The detailed balance condition is more than is needed to prove convergence, but a convenient restriction to impose. Popular examples of reversible algorithms are the Reversible Jump MCMC \citep{Green1995} and the Metropolis-Adjusted Langevin Algorithm \citep{Besag1993, Schreck2013}. Both techniques, however, depend on an accept/reject step in which proposed moves can be rejected, thereby wasting computational power. Birth-death (BD) algorithms are a popular reversible alternative. They operate in continuous time and explore the parameter space using repeated births and/or deaths, making them rejection-free. Recent examples of BD algorithms for Bayesian model inference include Gaussian graphical models \citep{Mohammadi2024} and Gaussian directed acyclic graphical models \citep{Jennings2016}. Despite being rejection-free, BD methods still suffer from poor mixing because they only allow local moves. 

Alternatively, non-reversible algorithms relax the detailed balance condition and employ different techniques to establish convergence to the stationary distribution. They enjoy better mixing properties and therefore explore large parameter spaces. In discrete time, a popular non-reversible MCMC strategy involves the idea of ``lifting" \citep{Sherlock2021,Gagnon2020}. In continuous-time, popular non-reversible algorithms are the Zig-Zag algorithm \citep{Bierkens2016} and the Bouncy Particle Sampler \citep{Bouchard2018}. Although both the Zig-Zag algorithm and the Bouncy Particle Sampler depend on differentiable parameter spaces, adaptations exist that allow for their application to binary model spaces, see for example \citep{Bierkens2022,Chevallier2023}. Despite better mixing properties, non-reversible algorithms still either require an accept/reject step or depend on local moves.

The starting point of our work is a reversible BD algorithm. We show that with a simple trick, we can transform such a process into a discrete-time MCMC algorithm that we call the Multiple Jump MCMC (MJ-MCMC) algorithm. Our algorithm is non-reversible and targets the marginal posterior. Unlike traditional accept/reject approaches, our algorithm is rejection-free, and unlike BD processes, our algorithm allows global updates. In fact,  the MJ-MCMC algorithm can traverse the entire model space in a single iteration. It is therefore exceptionally scalable, while maintaining high-quality inference. When applied to the field of undirected Gaussian graphical models (GGMs), our algorithm is $100$ to $200$ times faster than the state-of-the-art, solving $1000$-node problems (i.e. $2^{499500}$ potential models) in under $30$ seconds on a standard desktop. We show how to apply our approach to undirected GGMs, Ising models, and Bayesian variable selection, but we note that the applicability of our work reaches beyond those examples. In fact, it can be applied to most Bayesian posterior inference on binary model spaces. 

The outline of this article is as follows: Section \ref{sec:preliminaries} introduces the preliminaries of Bayesian model inference. We present our algorithm and a theoretical analysis in Section \ref{sec:algorithms}. We show how our work can be applied to undirected Gaussian graphical models, Bayesian variable selection, and Ising models in Section \ref{sec:applications}. We show the scalability of our work in a simulation study (Section \ref{sec:simulation}) and a real-life data application (Section \ref{sec:reallife}). We conclude and give directions for future research in Section \ref{sec:conclusion}.

\section{Preliminaries}
\label{sec:preliminaries}

\paragraph{Binary model spaces:} We consider unknown binary vectors $m = (m_1,\ldots,m_k) \in \mathcal{M} := \{0,1\}^k$ for some integer $k \geq 1$. That is, $m_i \in \{0,1\}$ for all $i=1,\ldots,k$ and $\mathcal{M}$ has $l := 2^k$ elements. We will refer to $m$ as a model and to $\mathcal{M}$ as the model space. With every model $m \in \mathcal{M}$, we associate a parameter vector $\theta_{m}$ in a parameter space $\bigTheta_{m}$. 




\paragraph{Bayesian model selection:} We assume that $m$ comes from a distribution $p(m)$ --- the model prior --- and that $\theta_{m}$ comes from a distribution $p(\theta_{m}  \mid  m)$ --- the prior on the parameter vector given the model. Given $m$ and $\theta_{m}$, we now observe a data vector $y$ from a distribution $p(y \mid m, \theta_{m})$ --- the likelihood. We can now formulate the joint posterior
\begin{equation}
p(m, \theta_{m}  \mid  y) = \frac{p(y  \mid  m, \theta_{m}) p(\theta_{m}  \mid  m)p(m)}{p(y)}.
\label{eq:joint_posterior}
\end{equation}
Here, $p(y) := \sum_{m} \int_{\bigTheta_{m}} p(y  \mid  m, \theta_{m}) p(\theta_{m}  \mid  m)p(m) \mathrm{d}\theta_{m}$
is the normalizing constant.

In many applications, one is more interested in the models $m$ than in the associated parameters $\theta_{m}$. In that case, it makes sense to integrate out the model parameters to obtain the marginal likelihood 
\begin{equation*}
    p(y  \mid  m) = \int_{\theta_{m} \in \bigTheta_{m}} p(y  \mid  m, \theta_{m}) p(\theta_{m}  \mid  m)~\mathrm{d}\theta_{m}.
\end{equation*}
The marginal likelihood is often referred to as the model evidence. The integral is sometimes available in closed-form -- mostly when selecting appropriate parameter priors $p(\theta_{m}  \mid  m)$. In most cases, however, the marginal likelihood needs to be approximated. Traditionally, the Bayesian Information Criterion \citep{Schwarz1978} and the Laplace method \citep{Tierney1986} are popular choices for this approximation. Recently, the pseudo-likelihood \citep{besag1975} has attracted renewed interest \citep{Mohammadi2024,Marsman2022}. With the marginal likelihood, one can formulate the model posterior
\begin{equation}
    p(m  \mid  y) = \frac{p(y  \mid  m) p(m)}{p(y)}.
\label{eq:marginal_posterior}
\end{equation} 
Bayesian model inference aims to understand either the joint posterior \eqref{eq:joint_posterior} or the model posterior \eqref{eq:marginal_posterior}. In many cases, both are challenging due to the normalizing constant $p(y)$. But even if the normalizing constant is available in closed form, computing the posterior for all $m \in \mathcal{M}$ is infeasible when the model space $\mathcal{M}$ is too large.

This is where MCMC methods come in useful. They circumvent the computation of the normalizing constant and are likely to compute the posterior only for those models that have a higher posterior probability. Moreover, MCMC methods come with a convenient property: they provide inference not only on the posterior, but also on any characteristic of the posterior. For example, the posterior probability that $m_i = 1$ for any $i =1,\ldots,k$, or the posterior probability that $\theta_{m} \in A$, for any set $A \subseteq \bigTheta_{m}$. MCMC methods come in two flavors: discrete-time MCMC methods and continuous-time MCMC methods.


\paragraph{Discrete-time Markov chains (DTMC):} A time-homogeneous DTMC is defined by the $l \times l$ row-stochastic transition probability matrix $P$ with elements $P(m, m')$ for all $m,m' \in \mathcal{M}$, where $\sum_{m'}P(m, m') = 1$. A DTMC explores the model space $\mathcal{M}$, each time moving from a model $m^{(s)}$ to a new model $m^{(s+1)}$ with probability $P(m^{(s)}, m^{(s+1)})$. A DTMC spends exactly one unit of time in each model before moving on to the next (with possible self-transitions). Hence, the classifier ``discrete time''. Under the right conditions, the chain $\{ m^{(s)} \}$ converges to a limiting distribution, denoted by $\pi(m)$. These conditions are irreducibility, aperiodicity, and the balance condition \citep[Theorem 1]{Tierney1994}. The balance condition is typically replaced by the detailed balance condition with respect to $\pi(m)$, that is,
\begin{equation}
    \pi(m)P(m, m')= \pi(m')P(m', m) \quad \forall m,m' \in \mathcal{M}.
    \label{eq:discrete_time_db}
\end{equation}
An MCMC chain that has detailed balance is called reversible. Though detailed balance is more restrictive than the general balance condition, it is typically convenient to impose. 

DTMCs are useful for Bayesian posterior inference, because one can set $P$ such that the detailed balance condition \eqref{eq:discrete_time_db} holds with respect to the model posterior $p(m  \mid  y)$. The resulting limiting distribution $\pi(m)$ then equals the posterior distribution $p(m  \mid  y)$ in \eqref{eq:marginal_posterior}, and the Markov chain samples $(m^{(1)},m^{(2)},\ldots)$, then converge to samples from $p(m  \mid  y)$. The Metropolis-Hastings algorithm by \citet{Hastings1970} and the Gibbs sampler by \citet{Geman1984} are two well-known examples of DTMCs where $P$ satisfies the detailed balance conditions.

\paragraph{Continuous-time Markov chains (CTMC):} A CTMC is defined by the $l \times l$ rate matrix $Q$ with elements $Q(m, m')$ for all $m,m' \in \mathcal{M}$ \citep{Anderson1991}. A CTMC explores the parameter space $\mathcal{M}$, each time moving from a state $m$ to a new state according to the following process. It samples, for every $m' \not = m$, a value from an exponential distribution with rate $Q(m, m')$. The chain then moves to the state $m'$ that corresponds to the lowest sampled value. This value is called the waiting time and represents the time the chain spends in each state before moving to the next. The waiting time is continuous and different for each state, hence the classifier ``continuous time''. Again, the detailed balance condition with respect to $\pi(m)$ is a convenient tool to design such a CTMC. It is expressed as
\begin{equation}
    \pi(m) Q(m,  m')= \pi(m')Q(m', m) \quad \forall m,m' \in \mathcal{M}.
     \label{eq:continuous_time_db}
\end{equation}
Note the similarity between \eqref{eq:discrete_time_db} and \eqref{eq:continuous_time_db}. Like a DTMC, we can design the rate matrix $Q$ such that detailed balance \eqref{eq:continuous_time_db} holds with respect to the model posterior $p(m  \mid  y)$ in \eqref{eq:marginal_posterior}. The resulting limiting distribution $\pi(m)$ then equals $p(m  \mid  y)$. Running the CTMC long enough will then result in samples from the posterior distribution.

Often, the size of the parameter space $\mathcal{M}$ is too large and computing the rate $Q(m, m')$ for all $m, m' \in \mathcal{M}$ is computationally infeasible. This is where so-called birth-death (BD) processes come in. BD processes are a class of CTMCs that set all rates $Q(m,  m')$ to zero, except when $m$ and $m'$ are ``neighbors''. Two states $m = (m_1,\ldots,m_k)$ and $m' = (m_1',\ldots,m_k')$ are neighbors when they differ in only one element $i \in \{1,\ldots,k\}$. Let the vector $m^i$ denote the state that differs from $m$ only in the $i$th element. Note that, since we deal with a binary model space, $m^i_i = 1 - m_i$. For a BD process, the detailed balance condition \eqref{eq:continuous_time_db} therefore simplifies to 
\begin{equation}
    \pi(m)Q(m, m^i) = \pi(m^{i})Q(m^i, m) \quad \forall m \in \mathcal{M}, i = 1,\ldots,k.
\label{eq:continuous_time_birth_death_db}
\end{equation}

A practical rejection-free birth-death sampler can be created as formalized in the following lemma.
\begin{lemma}
\label{lem:birth-death}
Let $\mathcal{M} = \{0,1\}^k$ be a binary model space. A CTMC with rates 
\begin{equation}
\label{eq:rates_birth_death}
Q(m, m') = \begin{cases} q_i(m), & \text{if } m' = m^i \text{ for some $i \in \{1,\cdots,k\}$};\\
0, & \text{otherwise},
\end{cases}
\end{equation}
where
\begin{equation}
\label{eq:rates_birth_death_q}
    q_i(m) = \min\left\{ 1,\frac{p(m^i \mid y)}{p(m \mid y)} \right\}
\end{equation} 
has the posterior $p(m \mid y)$ as its invariant distribution.
\end{lemma}
\begin{proof}
Using the rates \eqref{eq:rates_birth_death}, the detailed balance conditions \eqref{eq:continuous_time_birth_death_db} hold with respect to the model posterior $p(m \mid y)$. Hence, $p(m \mid y)$ is the invariant distribution of the CTMC.
\end{proof}

This BD process yields a simple and efficient algorithm. The main computational burden at each iteration is the computation of the $k$ rates $q_i(m)$. However, these rates can be computed in parallel. The MCMC chain will change \emph{exactly one element} of the model $m$ at each iteration, thereby offering a rejection-free alternative to the Metropolis-Hastings approach. It is applied to undirected GGMs \citep{Mohammadi2024} and to directed acyclic graphs \citep{Jennings2016}. Moreover, \citet{Chen2023}[Section 3.4] presents a framework containing the same idea and applies it to a low-scale Ising model. 

\paragraph{Bayes factor:} The ratio in \eqref{eq:rates_birth_death_q} can be written as:
\begin{equation*}
    \frac{p(m' \mid y)}{p(m \mid y)} = \frac{p(y  \mid  m')}{p(y  \mid  m)} \frac{p(m')}{p(m)}.
\end{equation*}
The first ratio on the right-hand side, $\frac{p(y  \mid  m')}{p(y  \mid  m)}$, is called the Bayes factor. It quantifies the evidence in favor of one model $m'$ over another model $m$. The Bayes factor is rarely available in closed form, because the marginal likelihoods it comprises aren't either; however, accurate closed-form approximations have appeared in many applications, including undirected Gaussian graphical models \citep{Leppaaho2017}, directed Gaussian graphical models \citep{Consonni2012}, Ising models \citep{Marsman2022}, and Bayesian variable selection \citep{Chipman2001}.

\section{Proposed Approach} 
\label{sec:algorithms}
A single iteration of the BD algorithm in Lemma \ref{lem:birth-death} is computationally feasible, because most rates are set to zero. However, the BD algorithm still struggles with large model spaces, because it updates just one element per iteration and is restricted to move to neighboring models only. Modern applications often treat models containing as many as one million elements, in which a BD algorithm needs millions of MCMC iterations, leading to computation and memory-related issues. This is why we now present an extension of the birth-death algorithm. This extension creates a discrete-time Markov chain with superior mixing qualities. In fact, the new DTMC can move across the entire model space in a single iteration. We first present the algorithm and show that it has two variants. Then, we provide a theorem for each variant, showing how accurately each variant targets the posterior distribution $p(m \mid y)$. We end this section with some additional remarks. Our algorithm provides samples from the marginal posterior \eqref{eq:marginal_posterior}, but we show that with a simple extension, it can generate samples from the joint posterior \eqref{eq:joint_posterior} too.

\subsection{Our Algorithm}
\label{sec:algorithm}
Consider a BD process with rates $Q(m, m')$ whose invariant distribution $\pi(m)$ is the posterior distribution $p(m \mid y)$ or an approximation thereof; see Lemma \ref{lem:birth-death} for an example of such a BD process. Our algorithm is a DTMC, derived from this BD process as follows. Let $\varepsilon_1,\varepsilon_2,\ldots,$ be a user-defined sequence with $\varepsilon_s \in (0,1)$ for $s=1,2,\ldots$. Also, for $m \in \mathcal{M}$ and $i \in \{1, \ldots, k\}$, let $q_i(m) := Q(m, m^i) \in (0,1]$. Suppose that in iteration $s$, the DTMC is in some state $m \in \mathcal{M}$. The DTMC moves to the next state by flipping its $i$th element with a probability of $q_i(m)\varepsilon_s$, independently for each $i=1,\ldots,k$. Hence, for arbitrary $m,m' \in \mathcal{M}$, the transition probability at iteration $s$ is given by 
\begin{equation} \label{eq:trans_prob}
    P_{\varepsilon_{s}}(m, m') = \prod_{i \in H_{m, m'}}q_i(m)\varepsilon_s \prod_{i \not \in H_{m, m'}}(1-q_i(m)\varepsilon_s),
\end{equation}
where $H_{m, m'} \subseteq \{1,\ldots,k\}$ is the set of elements in which $m$ and $m'$ differ. 
Note that the transition probability matrix $P_{\varepsilon_{s}}(m, m')$ is well-defined for any $\varepsilon_{s} \in (0, 1)$: we have $P_{\varepsilon_{s}}(m, m') > 0$  for all $m, m' \in \mathcal{M}$ and $\sum_{m'} P_{\varepsilon_{s}}(m, m') = 1$ for all $m \in \mathcal{M}$. The latter can be verified by observing that moving from a state $m$ to a new state $m'$ is equivalent to flipping $k$ different coins, each with a success probability of $q_i(m)\varepsilon_s$, $i=1,\ldots,k$. Hence, the sum of all the possible combinations of coin flips will sum to one.

 Algorithm \ref{alg:DTMC} provides the pseudo-code of our algorithm. As previously stated, Algorithm \ref{alg:DTMC} uses the rates of an existing BD process, and creates a new DTMC that explores the parameter space $\mathcal{M}$ faster. This is because it can update multiple elements in a single iteration. In fact, it can traverse the entire parameter space in a single iteration, because $P_{\varepsilon_{s}}(m, m') > 0$ provided that $\pi(m') > 0$ for any $m, m' \in \mathcal{M}$. In each state, the size of a jump (i.e. the number of elements $m_i$ that are updated) depends on the posterior probability of that state. Hence, we expect that the algorithm makes large jumps initially, but once it reaches an area with a high posterior probability, the jump size decreases, and the chain stabilizes.

\begin{algorithm}[H]
\label{alg:DTMC} 
\caption{Multiple Jump MCMC (MJ-MCMC)}
 \KwIn{ Data $y$, an initial model $m^{(1)}$ and a sequence $\varepsilon_1,\varepsilon_2,\ldots,\varepsilon_S$}
 \For{$s=1,\ldots,S$}{
   Calculate in parallel the rates $q_i(m^{(s)})$ for all $i=1,\ldots,k$ using \eqref{eq:rates_birth_death}\;
   Create $m^{(s+1)}$ by changing each element $m^{(s)}_i$ of $m^{(s)}$ with probability $q_i(m^{(s)})\varepsilon_s$\;
   }
 \KwOut{Samples $(m^{(1)},\ldots,m^{(S)})$} 
\end{algorithm}

Although it is relatively easy to see that Algorithm \ref{alg:DTMC} needs less iterations to explore the space $\mathcal{M}$, it is less obvious to assess how accurately the samples 
$(m^{(1)},\ldots,m^{(S)})$ target the posterior distribution $p(m  \mid  y)$. Section \ref{sec:theory} studies the theoretical properties of Algorithm \ref{alg:DTMC} concerning its accuracy. To that end, we consider two variants in terms of the sequence $\{\varepsilon_s\}_{s \geq 1}$. One variant uses a constant $\varepsilon_{s} = \varepsilon$,  which results in a homogeneous DTMC. A second variant uses a decaying $\{\varepsilon_s\}_{s \geq 1}$, resulting in an inhomogeneous DTMC. 

\subsection{Theoretical Support}

\label{sec:theory}
We investigate the behavior of our MJ-MCMC algorithm. Recall that the DTMC of our algorithm is built upon the reversible BD process with rate matrix $Q$ that has $\pi$ as its invariant distribution. In what follows, we will refer to this BD process as $(X(t))_{t \geq 0}$. We provide a theorem for both the homogeneous and inhomogeneous variants of Algorithm \ref{alg:DTMC}. 




\paragraph{Homogeneous case:}
In the first case, $\varepsilon_{s} = \varepsilon \in (0, 1)$ for all $s \geq 1$. To concretize the discussion, let $\{ M_{\varepsilon}^{(s)} \}_{s \geq 0}$ be the DTMC generated by Algorithm \ref{alg:DTMC} with constant $\varepsilon$. 
Moreover, define the following quantities regarding $(X(t))_{t \geq 0}$ and $\{ M_{\varepsilon}^{(s)} \}_{s \geq 1}$:
\begin{itemize}
\item The continuous-valued waiting time $W_{m}$ as the \emph{random time} the process $(X(t))_{t \geq 0}$ spends in state $m$ until it jumps to a new state,
\item The discrete-valued waiting time $W^{\varepsilon}_{m}$ as the \emph{random number of time steps} the process $\{ M_{\varepsilon}^{(s)} \}_{s \geq 1}$ spends in state $m$ until it jumps to a new state,
\item The \emph{jump probability} $\widetilde{P}(m, m')$ (resp.\ $\widetilde{P}_{\varepsilon}(m, m')$) as the probability that the process $(X(t))_{t \geq 0}$ (resp.\ $\{ M_{\varepsilon}^{(s)} \}_{s \geq 0}$) jumps from $m$ to $m'$ conditional on the event that a transition occurs. 
\end{itemize}
Theorem \ref{thm:homogeneous} below states several facts about the asymptotics of  $\{ M_{\varepsilon}^{(s)} \}_{s \geq 0}$ both in $\varepsilon$ and in time. Its proof is given in Supplementary Material A.2.



\begin{theorem}[Homogeneous DTMC]
\label{thm:homogeneous}
Let $\pi$ be the limiting distribution of the BD process $(X(t))_{t \geq 0}$ and $\pi(m) > 0$ for all $m \in \mathcal{M}$. Let $\pi_{\varepsilon}$ be the stationary distribution of $\{ M_{\varepsilon}^{(s)} \}_{s \geq 1}$. We have the following results:
\begin{enumerate}[label=(\roman*)]
\item $\pi_{\varepsilon}$ is the limiting distribution (in time) of $\{ M_{\varepsilon}^{(s)}\}_{s \geq 1}$.
\item If the BD process is reversible, then $||\pi_\varepsilon-\pi|| = O(\varepsilon)$ for any vector norm $|| \cdot||$.
\item For all states $m \in \mathcal{M}$, we have
\[
\varepsilon W_{m}^{\varepsilon} \overset{\text{d}}{\to} W_{m} \quad \text{as} \quad \varepsilon \to 0.
\]
That is, the waiting times of $\{M_{\varepsilon}^{(s)}\}_{s \geq 0}$, scaled by $\varepsilon$, converge, in distribution, to the waiting times of $( X(t))_{t \geq 0}$.
\item For all $m, m' \in \mathcal{M}$, we have
\[
\lim_{\varepsilon \to 0} \widetilde{P}_{\varepsilon}(m, m') = \widetilde{P}(m, m').
\]
That is, the jump probabilities of $\{ M_{\varepsilon}^{(s)}\}_{s \geq 0}$ converge to the jump probabilities of $(X(t))_{t \geq 0}$.
\end{enumerate}
\end{theorem}
Theorem \ref{thm:homogeneous} (i) and (ii) show that, as $\varepsilon \to 0$, the limiting distribution of the DTMC of Algorithm \ref{alg:DTMC} converges to $\pi$, the limiting distribution of the BD process. The convergence of (time-scaled) waiting times and jump probabilities, as shown in (iii) and (iv), shows that, as $\varepsilon$ gets smaller, Algorithm \ref{alg:DTMC} behaves more similarly to the BD process; however, at the expense of slowing down. Selecting $\varepsilon$ is therefore a trade-off between computational efficiency and accuracy. A higher value of $\varepsilon$ directly results in a faster exploration of the model space, since the probability that the element $i$ flips equals $q_i(m) \varepsilon$. However, lower values of $\varepsilon$ result in more proximity to the original BD process. Ultimately, the choice of $\varepsilon$ depends on the desired qualities of the resulting MCMC algorithm. In Section \ref{sec:simulation}, we compare the performance of Algorithm \ref{alg:DTMC} with different values of constant $\varepsilon$ in the case of Gaussian graphical models.



\paragraph{Inhomogeneous case:}
The next theorem treats the inhomogeneous case. Let $\{ M^{(s)} \}_{s \geq 0}$ be the DTMC created by the transition probabilities \eqref{eq:trans_prob} and with $\{\varepsilon_{s}\}_{s \geq 0}$. We show that the limiting distribution of $\{ M^{(s)}\}_{s \geq 0}$ is the same as the invariant distribution $\pi$ of the BD process, provided that the sequence $\{\varepsilon_s\}_{s \geq 1}$ vanishes and its sum  diverges. A proof of Theorem \ref{thm:inhomogeneous} is given in Supplementary Material A.3.

\begin{theorem}[Inhomogeneous DTMC]
\label{thm:inhomogeneous}
Consider a discrete-time MCMC over a binary model space $\mathcal{M} = \{0,1\}^k$ with transition probabilities \eqref{eq:trans_prob}, where $q_i(m)$ are the rates of a birth-death process, for all $m \in \mathcal{M}$ and $i=1,\cdots,k$. Moreover, assume that this birth-death process is irreducible and reversible. If the sequence $\{\varepsilon_s\}_{s \geq 1}$ satisfies 
\begin{equation}\label{eq: conditions for convergence}
\lim_{s \to \infty}\varepsilon_s = 0 \text{ (decay)}  \quad \text{ and} \quad \sum_{s=1}^{\infty}\varepsilon_s = \infty \text{ (divergence)},
\end{equation}
then the birth-death process and the DTMC have the same limiting distribution.
\end{theorem}

Theorem \ref{thm:inhomogeneous} requires decay and divergence of $\{ \varepsilon_{s} \}_{s \geq 1}$. That still leaves a wide variety of sequences to choose from. As in the homogeneous case, the choice of $\{ \varepsilon_{s} \}_{s \geq 1}$ presents a trade-off between computational efficiency and accuracy of the algorithm. A slowly decaying sequence (e.g., $\varepsilon_s = 1/\ln(s+2)$) mixes fast, but the chain might visit states with low posterior probability, resulting in poor accuracy. A quickly decaying sequence (e.g., $\varepsilon_s = 1/(s\ln s)$) may behave similarly to the original BD process -- accurate, but slow. 

Regardless of the chosen sequence, the consequences of Theorem \ref{thm:inhomogeneous}  are striking. It proves that the DTMC defined by transition probabilities \eqref{eq:trans_prob} has the same limiting distribution as the original birth-death process, namely $\pi$, but without the slow element-by-element updates. The DTMC, therefore, achieves the same quality of inference as the BD process with a remarkable reduction in computational cost. We demonstrate this claim with our thorough experiments (Section \ref{sec:simulation}) and an application in practice (Section \ref{sec:reallife}).

\subsection{Additional discussion and results}
\label{sec:extension}
This section builds on Theorem \ref{thm:homogeneous} and Theorem \ref{thm:inhomogeneous} to present some corollaries and remarks.

\paragraph{A notable quirk of Algorithm \ref{alg:DTMC}:} At every iteration $s$, our MJ-MCMC algorithm can jump to any state $m'$ with positive posterior probability, since $P_{\varepsilon_{s}}(m,  m') > 0$ whenever $\pi(m') > 0$. This comes with the evident advantage of excellent mixing. However, it also yields a noteworthy observation: our DTMC can visit states $m'$ with very low posterior probability, particularly for larger $\varepsilon_s$. Imagine for example the simple two-dimensional model space $\mathcal{M} = \{0,1\}^2$ and suppose that the posterior probabilities are $\pi(1, 1) = 0.01$, and $\pi(0, 0) = \pi(0, 1) = \pi(1, 0) = 0.33$. An MCMC algorithm targeting $\pi$ should therefore rarely visit the state $(1, 1)$. For example, the transition rates to $(1, 1)$ of the BD process in Lemma \ref{lem:birth-death} are $Q((0, 0), (1, 1)) = 0$ and $Q((0, 1), (1, 1)) = Q((1, 0), (1, 1)) = 0.03$. Our DTMC, however, is more likely to visit such a state, especially for higher $\varepsilon_s$.  For example, when in state $(0, 0)$ and when $\varepsilon_s = 0.9$, our DTMC has a chance of $P_{\varepsilon_s}((0, 0), (1, 1)) = \varepsilon_s^2 = 0.81$ of moving from $(0, 0)$ to $(1, 1)$. At first glance, this might seem problematic. However, note that Theorem \ref{thm:homogeneous} suggests a small value for the constant $\varepsilon$ and Theorem \ref{thm:inhomogeneous} requires $\lim_{s\ \to \infty} \varepsilon_s = 0$. These conditions ensure that, in the long run, jumps to states with low posterior probabilities become infrequent, and that the DTMC's limiting distribution will get arbitrarily close to that of the BD process. 

\paragraph{An exact (and slower) algorithm:}One could modify Algorithm \ref{alg:DTMC} with constant $\varepsilon$ to obtain an alternative (but slower) MCMC algorithm that exactly targets $\pi$. This alternative MCMC algorithm is an MH algorithm that uses $P_{\varepsilon_{s}}(m, m')$ \eqref{eq:trans_prob} to propose a new model $m'$ from the current state $m$. This proposed model $m'$ can then be accepted with a probability of 
\[
\alpha(m, m') := \min\left\{ 1, \frac{\pi(m') P_{\varepsilon_s}(m', m)}{\pi(m) P_{\varepsilon_s}(m,  m')} \right\},
\]
or rejected otherwise. The accept/reject step renders the algorithm exact, but requires extra computational cost.

\paragraph{A special case of convergence:} 
With homogeneous transition probabilities (i.e. $\varepsilon_s = \varepsilon$), the limiting distribution of Algorithm \ref{alg:DTMC} converges to the invariant distribution of the BD process as $\varepsilon \to 0$. There is a special case of Bayesian model inference, however, where both limiting distributions are the same for \textit{any} $\varepsilon \in (0,1)$. This special case requires the model posterior to be factorizable, i.e., we require
\begin{equation}
    p(m \mid y) = \prod_{i=1}^k p(m_i \mid y). \label{eq:fact_post}
\end{equation}
We can now state this result formally using the following lemma, a proof of which is given in Supplementary Material A.4.
\begin{lemma}
\label{lem:fact_post}
Consider a binary model space $\{0,1\}^k$ and a factorizable posterior $p(m \mid y)$, as in \eqref{eq:fact_post}. Consider a BD process as in Lemma \ref{lem:birth-death} with rates $q_i(m)$. Moreover, consider a homogeneous, discrete-time MCMC with transition probabilities \eqref{eq:trans_prob}, with $\varepsilon_s = \varepsilon \in (0,1)$ for all $s$, and where $q_i(m)$ are the rates of the BD process. Then, for any $\varepsilon \in (0,1)$, the DTMC has the same invariant distribution as the BD process.
\end{lemma}

Using a BD process or Algorithm \ref{alg:DTMC} to recover a factorizable posterior \eqref{eq:fact_post} is a form of over-engineering, since one can instead decompose the posterior into $k$ element-wise posteriors $p(m_i \mid y)$, each of which can be obtained easily using a two-state ($m_i = 0$ or $m_i = 1$) Markov chain. Lemma \ref{lem:fact_post} is therefore not meant as a proposal for a fast algorithm, but instead to provide intuition to the workings of Algorithm \ref{alg:DTMC}. The factorizable posterior ensures that the elements $m_1,\cdots,m_k$ are independent. This element-wise posterior independence is exactly why the homogeneous variant of Algorithm \ref{alg:DTMC} works for all $\varepsilon \in (0,1)$, since Algorithm \ref{alg:DTMC} updates the elements $m_i$ independent of each other. In most cases, however, posteriors are not factorizable. That is why, in the general case, we need either the condition $\varepsilon \to 0$ (Theorem \ref{thm:homogeneous}) or both decay and divergence conditions (Theorem \ref{thm:inhomogeneous}).\\

\noindent \textbf{Extension to joint posterior:} So far, we have presented an extension to the BD process (Lemma \ref{lem:birth-death}) that creates a superior DTMC (Algorithm \ref{alg:DTMC}). Both the BD process and our DTMC are suitable for inference on the model posterior \eqref{eq:marginal_posterior}, but not to learn the parameters $\theta_{m}$ associated with each model. That is, the algorithms are not immediately applicable for inference on the joint posterior \eqref{eq:joint_posterior}. The step to joint inference, however, is straightforward. In fact, the next lemma presents a simple strategy to use the samples $(m^{(1)},\ldots,m^{(S)})$ from Algorithm \ref{alg:DTMC} to obtain samples $((m^{(1)},\theta^{(1)}_{m}),\ldots,(m^{(S)},\theta^{(S)}_{m})$ from the joint posterior \eqref{eq:joint_posterior}.

\begin{lemma}
\label{lem:sample_ks}
Let $m^{(1)},m^{(2)},\ldots$ be the states of a discrete-time Markov chain with limiting distribution $\pi(m)$. If we now sample the parameter vectors $\theta^{(1)},\theta^{(2)},\ldots$ from a distribution $\pi(\theta  \mid  m^{(s)})$ for all $s=1,2,..$. Then, the resulting samples $(m^{(1)},\theta^{(1)}), (m^{(2)},\theta^{(2)}), \ldots$ are samples from a Markov chain with limiting distribution $\pi(m)\pi(\theta  \mid  m)$.
\end{lemma}
\begin{proof}
Let $\mu_s(m)$ denote the probability that the Markov chain with states $m^{(1)},m^{(2)},\ldots$ is in state $m$ at iteration $s$. We have $\lim_{s \to \infty} \mu_s(m) = \pi(m)$. Let $\mu_s(m,\theta)$ denote the probability that the Markov chain with states $(m^{(1)},\theta^{(1)}), (m^{(2)},\theta^{(2)}), \ldots$, is in state $(m,\theta)$ in iteration $s$. Then, we obtain
\begin{equation*}
\lim_{s \to \infty}\mu_s(m,\theta) = \lim_{s \to \infty}\mu_s(m)\pi(\theta  \mid  m) = \pi(m)\pi(\theta  \mid  m).
\end{equation*}
Hence, the limiting distribution of the Markov chain with states $(m^{(1)},\theta^{(1)}), (m^{(2)},\theta^{(2)}), \ldots$ is $\pi(m)\pi(\theta  \mid  m)$.
\end{proof}
 \noindent If $\pi(m) = p(m  \mid  y)$ and $\pi(\theta  \mid  m) = p(\theta  \mid  m, y)$, then the Markov chain with states $(m^{(1)},\theta^{(1)}), (m^{(2)},\theta^{(2)}),\ldots$ has limiting distribution $\pi(m)\pi(\theta  \mid  m) = p(m  \mid  y)p(\theta  \mid  m, y) = p(\theta,m  \mid  y)$, which is exactly the required joint posterior.

\section{Applications}
\label{sec:applications}
The MJ-MCMC algorithm (Algorithm \ref{alg:DTMC}) can be used for Bayesian model inference on problems with a binary model space. In Supplementary Material B, we apply our approach to Bayesian variable selection and Ising models. In this section, we present a different example: undirected Gaussian graphical models (GGMs). GGMs are a good starting point, because it is a relatively well-known field. In fact, the BD process in Lemma \ref{lem:birth-death} is already proposed by \citet{Mohammadi2024}. This subsection reviews those existing results and presents our algorithm as an extension of this BD process.


\paragraph{Preliminaries:} GGMs assume that the data $y$ come from a multivariate normal distribution \citep{lauritzen1996}. More precisely, $y$ is an $n \times p$ matrix containing $n$ observations of the variables $Y_1,\ldots,Y_p$. These variables follow a multivariate normal distribution with unknown $p \times p$ covariance matrix $\Sigma$ and unknown precision matrix $K := \Sigma^{-1}$. The matrix $K$, with elements $k_{ij}$, encodes the conditional dependency between any two variables $Y_i$ and $Y_j$, because 
\begin{equation*}
    k_{ij} = 0 \iff Y_i \text{ and } Y_j \text{ are conditionally independent}. 
\end{equation*}
These conditional (in)dependence relationships can be depicted in a graph $G$ with $p$ nodes corresponding to $Y_1,\ldots,Y_p$. An undirected edge $(i,j)$ is included if and only if the variables $Y_i$ and $Y_j$ are conditionally dependent. We are interested in the graph $G$ that could have generated these data $y$. Bayesian inference focuses on uncovering the posterior distribution of this graph. That is, 
\begin{equation*}
  p(G  \mid  y) = \frac{p(y  \mid  G ) p(G)}{p(y)}.
\end{equation*}
If we denote the presence of an edge with a $1$ and the absence of an edge with a 0, we can encode every graph as a binary vector of length $k=p(p-1)/2$. The graphs $G$ are therefore binary ``models'', and the space of all $p \times p$ graphs is a binary model space. The precision matrix $K$ is the model-associated parameter. Given a graph $G$, the prior on the precision matrix is denoted by $p(K  \mid  G)$. Popular choices include the G-wishart prior \citep{roverato2002hyper} or the spike-and-slab prior \citep{wang2015}. See \citet{vogels2024} for a comprehensive review on Bayesian structure learning in GGMs.

\paragraph{BD process:} Let $G$ and $G^e$ be two graphs that differ only in the edge $e = (i,j)$. Under the G-Wishart prior, an approximation of the Bayes factor between $G$ and $G^e$ exists \citep{Leppaaho2017}. It is given by 
\begin{equation}
\label{eq:GGM_1}
\frac{p(y \mid G^e)}{p(y \mid G)} \approx \frac{\hat{p}(y \mid G^e)}{\hat{p}(y \mid G)} := \frac{p(y_i \mid y_{nb(i,G^e)})
p(y_j \mid y_{nb(j, G^e)})}{p(y_i \mid y_{nb(i,G)}) p(y_j \mid y_{nb(j,G)})},
\end{equation}
where $y_h$ is a vector containing the $n$ observations of variable $h$, $nb(h,G)$ refers to the set of neighbors of node $h$ with respect to $G$ and $y_{A}$ is the sub-matrix of $y$ corresponding to the nodes that are in a set $A$. This approximation is remarkably cheap to evaluate, because it only requires the evaluation of the four terms $p(y_h \mid y_{nb(h,G)})$, each of which can be obtained quickly using another approximation \citep{Leppaaho2017}[Equation 9]. \citet{Mohammadi2024} use equation \eqref{eq:GGM_1} to design a BD process as described in Lemma \ref{lem:birth-death} with rates 

\begin{equation*}
Q(G, G') = \begin{cases} q_e(G), & \text{if } G' = G^e \text{ for some edge $e$};\\
0, & \text{otherwise,}
\end{cases}
\end{equation*}
where 
\begin{equation}
q_e(G) := \min \left\{1,\frac{\hat{p}(y \mid G^e)}{\hat{p}(y \mid G)} \frac{p(G^e)}{p(G)} \right\}.
\label{eq:rates_GGM}
\end{equation}

The algorithm of \citet{Mohammadi2024} outperforms other state-of-the-art algorithms, but still takes more than an hour to solve $1000$-variable problems (see Section \ref{sec:simulation}), mostly because the algorithm -- being a BD process -- updates only one edge per iteration.

\paragraph{Our algorithm:} Our work allows us to extend the BD process of \citet{Mohammadi2024} to a DTMC that does allow multiple edge updates per iteration. Our DTMC will have transition probabilities \eqref{eq:trans_prob} with $q_i(m) := q_e(G)$ as in \eqref{eq:rates_GGM}. By Theorem \ref{thm:inhomogeneous}, the DTMC has the same limiting distribution as the BD process, provided that the sequence $\{\varepsilon_{s}\}_{s \geq 1}$ satisfies \eqref{eq: conditions for convergence}. 
Moreover, similar to \citep{Mohammadi2024}[Sectiopn 4.3], we can sample precision matrices $K^{(s)}$ from $p(K  \mid  G^{(s)}, y)$ for $s=1,\ldots,S$. Due to Lemma \ref{lem:sample_ks}, the resulting samples $(G^{(1)},K^{(1)}),\ldots,(G^{(S)},K^{(S)})$ are then samples from the joint posterior $p(K, G  \mid  y)$. We tested our approach in a simulation study (Section \ref{sec:simulation}) and on real-life data (Section \ref{sec:reallife}). 

\section{Numerical Experiments}\label{sec: Numerical Experiments}

This section demonstrates the performance of the MJ-MCMC approach in practice. We apply our algorithm to the field of undirected Gaussian graphical models, using a simulation study (Section \ref{sec:simulation}) and a real-life data set (Section \ref{sec:reallife}). 

\subsection{Simulation study}
\label{sec:simulation}

This section highlights the scalability of our approach in a simulation setting. We compare our approach (Algorithm \ref{alg:DTMC}) to the state-of-the-art algorithm in Bayesian structure learning in Gaussian graphical models: the BD-MPL algorithm \citep{Mohammadi2024}, which is up to $10$ times faster than the competing algorithms on large-scale instances ($p=1000$) \citep{vogels2024}. We test our algorithm with homogeneous transition probabilities, using $\varepsilon \in \{0.001, 0.01, 0.3, 0.7, 0.95\}$, and inhomogeneous transition probabilities \eqref{eq:trans_prob} where we use the slowly and quickly decaying sequences
\begin{equation}
\varepsilon_s = 0.3\frac{1}{\log_{10}(s+9)} \quad \text{ and } \quad \varepsilon_s = 0.3\left(\frac{1}{s\log_{2}(s+1)}\right)^{0.4} \text{ for } s=1,2\ldots
\label{eq:sequence_decay}
\end{equation}
respectively. Figure \ref{fig:decaying} shows both sequences. We report our results on an instance with $p=1000$ variables, $n=400$ observations, and edge density of $0.2\%$. This edge density might seem low, but note that under this density, nodes are connected to two other nodes on average. The experiments were repeated eight times. Other simulation settings are described in Supplementary Material C.

\begin{figure}[!ht] 
    \centering
\includegraphics[width=0.5\linewidth]{ 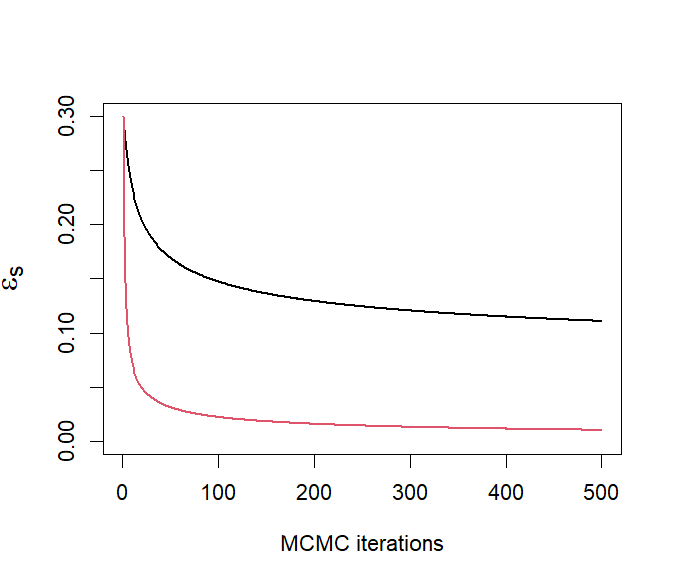}
\caption{\textit{The decay during the first $500$ iterations for a quickly (red) and slowly (black) decaying sequence.
}}
\label{fig:decaying} 
\end{figure}

For each algorithm, we use the generated Markov Chain to calculate the posterior edge inclusion probability. This value represents the probability that, given the data, an edge $(i,j)$ is included in the graph. We then assess the quality of these edge inclusion probabilities using the AUC-PR, i.e., the area under the precision-recall curve. The AUC-PR ranges from zero (worst) to one (best).

\begin{figure}[!ht] 
    \centering
    \begin{tabular}[t]{cc}
        \begin{subfigure}[t]{0.5\textwidth}
            \centering            \includegraphics[width=1\linewidth]{ 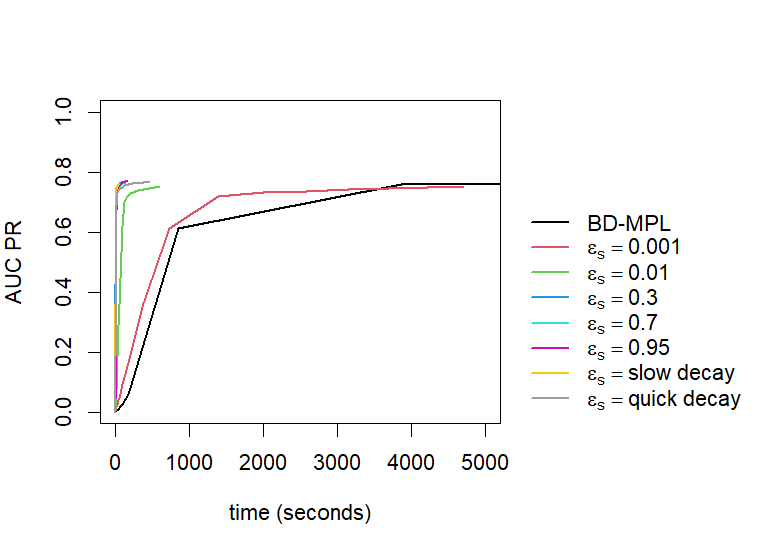}
        \end{subfigure} &
        \begin{subfigure}[t]{0.5\textwidth}
            \centering            \includegraphics[width=1\textwidth]{ 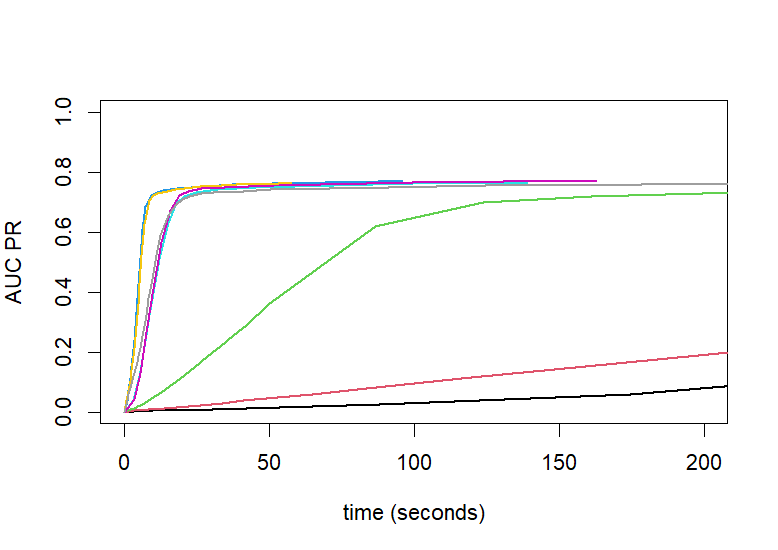}
        \end{subfigure}
    \end{tabular}\\
\caption{\textit{AUC-PR scores over running time for the BD-MPL algorithm as well as for our approach with homogeneous and inhomogeneous transition probabilities. Zoomed out on the left, zoomed in on the right. Results are the average of eight replications on instances with $p=1000$ variables, $n=400$ observations, and edge density of $0.2\%$.
}}
\label{fig:p1000_n400_sparse} 
\end{figure}

Figure \ref{fig:p1000_n400_sparse} contains the results. It leads to three main conclusions: First, our approach is remarkably fast. In the homogeneous case for $\varepsilon_s  \in \{0.3, 0.7, 0.95\}$, our approach solves $1000$-variable problems in under $30$ seconds. That is around $100$ to $200$ times faster than the state-of-the-art BD-MPL algorithm. In the inhomogeneous case, the quickly and the slowly decaying sequences are even faster. Second, for small values of $\varepsilon$, i.e., $\varepsilon \in \{0.001,0.01\}$, our approach slows down and starts behaving like the BD-MPL algorithm. This is in line with Theorem \ref{thm:homogeneous}. Third, higher values of $\varepsilon$ do not lead to worse accuracy. On the contrary, all values of $\varepsilon$ as well as the quickly and the slowly decaying sequences lead to similar AUC-PR values. Lemma \ref{lem:fact_post} might provide a reason for this surprising result. Using the G-wishart prior, the resulting posterior is not factorizable as in \eqref{eq:fact_post}. However, for large values of $p$ and for sparse instances, we hypothesize that it does behave approximately factorizable. In line with Lemma \ref{lem:fact_post}, this ensures that the limiting distribution of Algorithm \ref{alg:DTMC} is similar to that of the BD-MPL approach for all values of $\varepsilon$. This hypothesis is supported by Figures S9-S16 in Supplemenary Material C.3 and C.4. These figures show that, in denser cases, high values of $\varepsilon$ (e.g. $\varepsilon = 0.75$) do in fact lead to worse accuracy. 

This section only included one instance and one metric. In Supplementary Material C we provide results for different metrics, for a small and large number of observations, and for sparse and dense graphs. These supplementary results tell the same story: across metrics and across instances, Algorithm \ref{alg:DTMC} reduces computational efficiency without sacrificing accuracy. 

\subsection{Real-life data}
\label{sec:reallife}
The simulation study shows that our approach offers a significant reduction in computation time compared to the state-of-the-art. This section will show that this also holds for real-life data. Specifically, we will test Algorithm \ref{alg:DTMC} on a dataset with $p=623$ variables and $n=653$ observations, and show that we obtain the same quality of inference as the state-of-the-art BD-MPL algorithm \citep{Mohammadi2024}, but up to $10$ times as fast. 

The data set in question is the GSE15907 microarray dataset from \citet{Painter2011} and \citet{Desch2011}, comprising gene expression data from 24,922 genes in 653 mouse immune cells, obtained from the Immunological Genome Project \citep{ImmGen}. Studying the conditional dependency structure among genes is important as it hints at how genes in immune cells work together to produce proteins impacting the body's response to diseases. Understanding these mechanisms is essential to improving disease treatment strategies. The dataset was used in the context of Bayesian inference in GGMs before by both \citet{chandra2024} and \citet{Mohammadi2024}. 

For data preparation, we follow the steps of \citet{chandra2024}. That is, we apply a $\log_2$ transformation to avoid heteroskedasticity, and we only keep the 2.5\% of genes with the highest variance. The final dataset, therefore, contains $p = 2.5\% \times 24922 = 623$ genes. Each of the $n=653$ observations corresponds to an immune cell. Cells belonging to the same cell type have similar gene expression profiles, whereas cells belonging to different cell types behave very differently. We therefore center, for each gene, the observations within each cell type.\footnote{\citet{Mohammadi2024} skip this step. The posterior they report is therefore denser than ours} This also ensures that each variable has a zero mean. Lastly, we transform the data in line with \citep{Liu2009}. This transformation ensures that each marginal distribution can be assumed to be a Gaussian with unit variance.

We use the Bernoulli prior for the graph with a prior edge inclusion probability of $0.01$, based on the results of \citet{chandra2024}. We will run Algorithm \ref{alg:DTMC} with homogeneous transition probabilities $\varepsilon \in \{0.001,0.01,0.3,0.6,0.9\}$ and inhomogeneous transition probabilities, where we use the slowly and quickly decaying sequences \eqref{eq:sequence_decay}. We compare the performance of our algorithm with the BD-MPL algorithm by \citet{Mohammadi2024}. We initialize both algorithms at the empty graph. 

Our algorithm provides the same quality of Bayesian inference as the BD-MPL algorithm. Figure \ref{fig:scatter} provides a scatterplot comparing the estimated edge inclusion probabilities of Algorithm \ref{alg:DTMC} ($\varepsilon = 0.3$) with those of the BD-MPL algorithm. Almost all dots are close to the red diagonal line, indicating that the inference of our algorithm is very similar to that of the BD-MPL algorithm. Similar results apply for other values of $\varepsilon$, including the slowly and quickly decaying sequences. In fact, for every $\varepsilon$, the correlation coefficient between the edge inclusion probabilities of Algorithm \ref{alg:DTMC} and those of the BD-MPL algorithm is higher than $0.988$ and the average absolute difference smaller than $0.004$. An exception is $\varepsilon=0.9$, in which case the MJ-MCMC algorithm jumps uncontrollably and ends up at a different target distribution altogether.

\begin{figure}[!ht] 
    \centering
\includegraphics[width=0.5\linewidth]{ 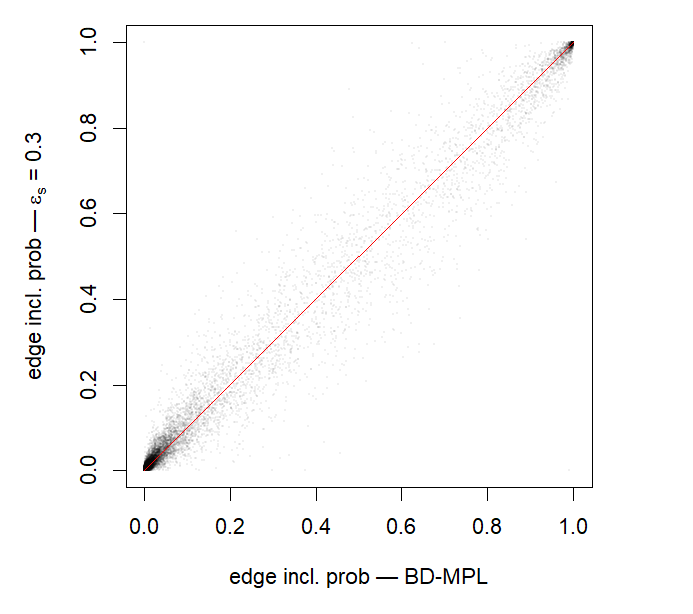}
\caption{\textit{Scatterplot comparing the edge inclusion probabilities of the BD-MPL algorithm with our MJ-MCMC algorithm with $\varepsilon=0.3$.
}}
\label{fig:scatter} 
\end{figure}

Figure \ref{fig:running_times} shows the total running time and total MCMC iterations required for each algorithm. The MJ-MCMC algorithm is up to $10$ times faster than the BD-MPL algorithm, and requires up to $600$ times fewer MCMC iterations. In line with Theorem \ref{thm:homogeneous}, we observe that as $\varepsilon \to 0$, our algorithm becomes slower and starts behaving like the BD-MPL algorithm. To create Figure \ref{fig:running_times}, we run each algorithm until a stopping criterion is met. To determine this criterion, we first run the BD-MPL algorithm for four million MCMC iterations, more than sufficient for convergence. Next, we run each algorithm until the resulting edge inclusion probabilities have a Pearson correlation of $0.98$ with those of the BD-MPL algorithm. At that point, we report the total running time and MCMC iterations. Note that, when $\varepsilon = 0.9$, this point is never reached and is therefore noted with a \textit{NA} in Figure \ref{fig:running_times}.

\begin{figure}[htbp]
\centering
\includegraphics[width=0.48\textwidth]{ 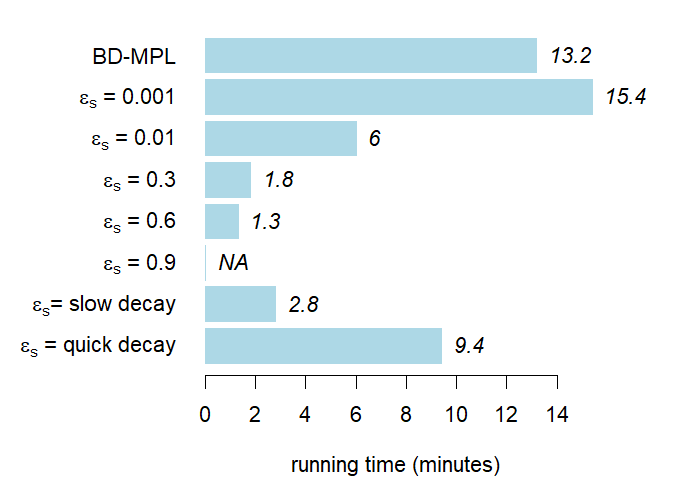}
\hfill
\includegraphics[width=0.48\textwidth]{ 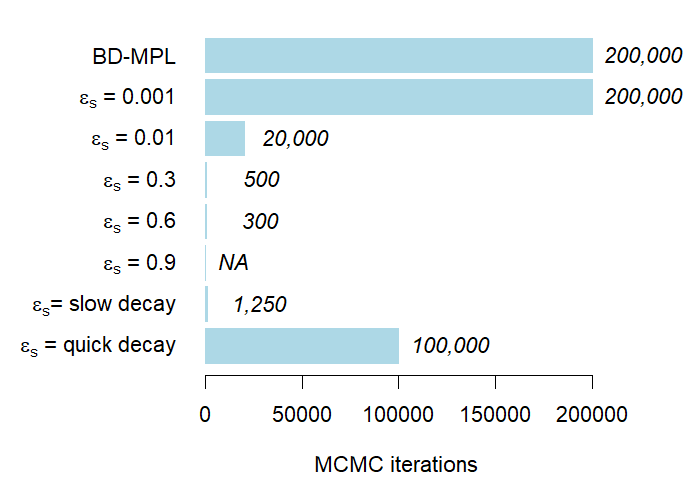}
\caption{The running time (left) and required MCMC iterations (right) of the BD-MPL algorithm compared to Algorithm \ref{alg:DTMC} for different values of $\varepsilon_s$, on the mice gene dataset with $p = 623$.}
\label{fig:running_times}
\end{figure}

\section{Conclusion}
\label{sec:conclusion}
This work introduces an algorithm for Bayesian model inference on binary model spaces. We take a BD process as a starting point and show that a simple extension creates a new DTMC. Under mild conditions, this DTMC has the same limiting distribution as the BD process. However, it can cross the entire model space in a single iteration, leading to remarkable computational efficiency. When applied to Gaussian graphical models, our algorithm is up to $200$ times as fast as the state-of-the-art. We show that our framework can be applied to Ising models and variable selection, too. There are several ways in which future work can build on this article.


First, our framework can be applied beyond the examples given in this paper. Starting with other binary models, such as mixed graphical models \citep{Florez2025}, decision trees \citep{Mohammadi2020}, mixture distributions \citep{Stephens2000} or directed graphical models \citep{Kuipers2022}. Moreover, every discrete model space can be written as a binary model space with a constraint. This opens up the possibility of applying our work to Bayesian inference on discrete model spaces. Lastly, \citet{Chen2023}[Section 4.4] show that an algorithm designed for a discrete state space can be extended to work on a continuous state space. A similar extension could apply our framework to continuous model spaces. 

Second, our framework can be extended for inference on the joint posterior. Algorithm \ref{alg:DTMC} focuses on inference on the model posterior \eqref{eq:marginal_posterior}. There are two ways to realize this extension. First, we can follow  Lemma \ref{lem:sample_ks} and sample the model-associated parameters after running Algorithm \ref{alg:DTMC}. Simulations are needed to show whether this works in practice. Alternatively, we can take a BD process on the joint space as a starting point (e.g. \citet{Mohammadi2015}), and tweak Algorithm \ref{alg:DTMC} to extend this to a superior DTMC on the joint space. A theoretical analysis and simulation study are needed to assess this strategy.

Third, further reductions in computational time are possible. One such strategy could be to update only a subset of the rates at each iteration. \citet{Chen2023}[Section 4.2] proposes this strategy in a similar set-up, and shows its potential to reduce computational costs. Alternatively, we can achieve higher efficiency by calculating the rates using only a subset of the observations, as in \citet{Welling2011}. 

Lastly, there is a lot of room for more theoretical guarantees. Although Theorem \ref{thm:homogeneous} and \ref{thm:inhomogeneous} shed light on the accuracy of the algorithm in terms of the proximity (or convergence) to the posterior distribution, we provide no results on the rate of this proximity (convergence). Moreover, \citet{Andrieu2009} show that an unbiased estimate of the likelihood is enough to design a Metropolis-Hastings algorithm that targets the exact posterior. There is potential to extend this result to a BD process and to the framework discussed in this article.




\begin{supplement}
\stitle{\textbf{Supplement A.1}}
\sdescription{Preliminary lemmas used in the proofs of Theorems \ref{thm:homogeneous} and \ref{thm:inhomogeneous}}
\end{supplement}

\begin{supplement}
\stitle{\textbf{Supplement A.2}}
\sdescription{The proof of Theorem \ref{thm:homogeneous}}
\end{supplement}

\begin{supplement}
\stitle{\textbf{Supplement A.3}}
\sdescription{The proof of Theorem \ref{thm:inhomogeneous}}
\end{supplement}

\begin{supplement}
\stitle{\textbf{Supplement A.4}}
\sdescription{The proof of Lemma \ref{lem:fact_post}}
\end{supplement}

\begin{supplement}
\stitle{\textbf{Supplement B.1}}
\sdescription{The application of MJ-MCMC to Bayesian variable selection}
\end{supplement}

\begin{supplement}
\stitle{\textbf{Supplement B.2}}
\sdescription{The application of MJ-MCMC to structure selection in Ising models}
\end{supplement}

\begin{supplement}
\stitle{\textbf{Supplement C}}
\sdescription{A detailed description of all instances, settings, metrics, and results of the simulation study of Section \ref{sec:simulation}.}
\end{supplement}

\begin{supplement}
\stitle{\textbf{Code}}
\sdescription{All code used to create the figures of this article is available on  \url{https://github.com/lucasvogels33/A-Scalable-MCMC-Algorithm-For-Bayesian-Inference-on-Binary-Model-Spaces}}
\end{supplement}

\bibliographystyle{ba}
\bibliography{sample}

\end{document}


\begin{frontmatter}
\title{Supplementary Material \\ for \\ Multiple Jump MCMC: a Scalable Algorithm for Bayesian Inference on Binary Model Spaces}
\runtitle{Scalable MCMC algorithm for Bayesian inference}

\begin{aug}

\end{aug}


\end{frontmatter}

\section{Proofs of the theoretical results in Section 3} \label{proof_of_theorem_Ttozero}

\subsection{Preliminary Lemmas} \label{sup: Preliminary lemmas}
Recall that $Q$ is the rate matrix of the BD process with transition rates $q_{i}(m)$ and limiting distribution $\pi$, with $\pi(m) > 0$ for all $m \in \mathcal{M}$. Moreover, let $P_\varepsilon$ denote the transition probability matrix corresponding to the homogeneous DTMC generated by Algorithm 1 in Section 3.1. That is,  $P_{\varepsilon}$ is the matrix with entries $P_{\varepsilon}(m,m')$ defined in Equation (8), with $\varepsilon_s = \varepsilon$. The following lemmas will be used in the proofs of Theorems 3.1 and 3.2.

\begin{lemma}\label{lem: P in terms of I and Q}
For any $\varepsilon > 0$, we have
\begin{equation} \label{eq:homo_step1}
P_{\varepsilon} = I + \varepsilon Q + E_{\varepsilon},
\end{equation}
    where $E_{\varepsilon}$ is an error matrix with entries of order $\mathcal{O}(\varepsilon^2)$.
        \end{lemma}
        \begin{proof} First, for $m \in \mathcal{M}$, we have 
    \begin{align}
    P_{\varepsilon}(m,m) & = \prod_{i=1}^{k}(1-Q(m, m^i)\varepsilon) \nonumber \\
    & = 1 - \varepsilon \sum_{i=1}^k Q(m, m^i) + \mathcal{O}(\varepsilon^2) \nonumber\\
    & = 1 + \varepsilon Q(m,m) + \mathcal{O}(\varepsilon^2), \label{eq:inhomo_one}
    \end{align}
    where we use the Taylor expansion of the product term, and in the last line, we use the fact that $Q(m, m) = -\sum_{i = 1}^{k} Q(m, m^i)$. Next, for model pairs $(m, m^{i})$ for any $i = 1, \ldots, k$, we have
\begin{align}
P_{\varepsilon}(m,m^i) & = Q(m, m^i)\varepsilon \prod_{j \not = i}^{k}(1-Q(m,m^j)\varepsilon)\nonumber\\
& = Q(m, m^i)\varepsilon \left(1 - \varepsilon \sum_{j \not = i}Q(m, m^j) + \mathcal{O}(\varepsilon^2) \right) \nonumber\\
& = \varepsilon Q(m,m^i) + \mathcal{O}(\varepsilon^2). \label{eq:inhomo_two0}
\end{align}
Lastly, for other model pairs $(m,m')$, we use the fact that $Q(m,m')$ = 0. We then have
\begin{equation}
P_{\varepsilon}(m,m') = \mathcal{O}(\varepsilon^2) = \varepsilon Q(m,m') + \mathcal{O}(\varepsilon^2).\label{eq:inhomo_three} 
\end{equation}
Combining \eqref{eq:inhomo_one}-\eqref{eq:inhomo_three}, we obtain \eqref{eq:homo_step1}.
\end{proof}

For two $l$-dimensional vectors $x,x'$ with elements $x(m)$ and $x'(m)$, $m \in \mathcal{M}$, let us define $\langle x,x'\rangle_\pi := \sum_mx(m)x'(m)\pi(m)$. Let $||x||^2_{2,\pi} = \langle x,x\rangle_\pi$. Note that this norm is well-defined since $\pi(m) > 0$ for all $m \in \mathcal{M}$, which is ensured by the irreducibility of the chain $Q$ and the finiteness of $\mathcal{M}$.

\begin{lemma}
\label{lem:inhomo}
If $Q$ is reversible with regard to $\pi$, there exist constants $\beta > 0$ and $\varepsilon^* > 0$ such that for all $\varepsilon < \varepsilon^*$ and any row vector $e \in \mathbb{R}^{1 \times l}$ with $\sum_{i = 1}^{l} e_{i}= 0$, we have
\begin{equation}|| e (I + \varepsilon Q)||_{2,\pi} \leq (1-\varepsilon \beta)||e ||_{2,\pi}.
\label{eq:inhomo_seven}
\end{equation}
\end{lemma}
\begin{proof}
This is a known result, and, in a more general form, is referred to as the Poincaré inequality. For completeness, and because a step-by-step proof is missing in the literature, we provide the proof here. First, we observe three critical facts about the matrix $\widehat{P}_{\varepsilon} := I+\varepsilon Q$:
    \begin{itemize}
    \item Let $0 = \beta_1 > \beta_2 \geq \beta_3 \geq \ldots$ be the eigenvalues of $Q$, where $\beta_1 = 0$ is ensured by the fact that $Q$ is a valid rate matrix and  $\beta_2 < 0$ is guaranteed by the fact that $Q$ is convergent. Then, for $\varepsilon > 0$, the matrix $\widehat{P}_{\varepsilon}$ has eigenvalues $\lambda_i := 1 + \varepsilon \beta_{i}$, for $i = 1, \ldots, k$, with the same set of eigenvectors, which we denote by $v_i$. In particular, the largest eigenvalues of $\widehat{P}_{\varepsilon}$ is $\lambda_1 := 1$ with corresponding eigenvector $v_1 := \mathbf{1}$. Moreover, the second largest eigenvalue of $\widehat{P}_{\varepsilon}$ is
    $\lambda_{2} := 1 + \beta_{2} \varepsilon$.
    
        \item There exists an $\varepsilon_{1}^* > 0$ such that for all $\varepsilon < \varepsilon_{1}^*$ the matrix $\widehat{P}_{\varepsilon}$ is a valid transition probability matrix that is reversible with respect to $\pi$. The reversibility is inherited from the reversibility of $Q$. Since $\widehat{P}_{\varepsilon}$ is reversible with respect to $\pi$, the eigenvectors $v_i$ form an orthonormal basis with inner product $\langle \cdot, \cdot\rangle_\pi$. That is, $\langle v_i,v_i\rangle_\pi = 1$, $\langle v_i,v_j\rangle_\pi = 0$ for all $i \not = j$, and any vector $f$ can be written as $f = \sum_{i=1}^lc_iv_i$ with $c_i := \langle f,v_i \rangle_\pi \in \mathbb{R}$.
    \item  Furthermore, pick some $\varepsilon_{2}^{\ast} < -1/\beta_{2}$ and let $\varepsilon^{\ast} = \min\{ \varepsilon_{1}^*, \varepsilon_{2}^*\}$. Then, for any $\varepsilon < \varepsilon^{\ast}$, $\widehat{P}_{\varepsilon}$ is a valid transition probability matrix whose second eigenvalue $\lambda_{2} = 1 + \varepsilon \beta_{2} \in (0, 1)$.
\end{itemize}
Take any $\varepsilon < \varepsilon^{\ast}$ and let $f = \sum_{i=1}^lc_i v_i$ be such that $c_1 := \langle f,\mathbf{{1}} \rangle_\pi = 0$. We have
\begin{equation*}
f \widehat{P}_{\varepsilon} = \sum_{i = 1}^lc_i  v_i \widehat{P}_{\varepsilon} = \sum_{i = 1}^lc_i \lambda_i v_i  = \sum_{i = 2}^lc_i \lambda_i v_i.
\end{equation*}
Therefore, we have
\begin{align*}
    ||f \widehat{P}_{\varepsilon} ||^2_{2,\pi} & = \langle f \widehat{P}_{\varepsilon},f \widehat{P}_{\varepsilon} \rangle_\pi \\
    & = 
    \left\langle \sum_{i = 2}^lc_i \lambda_i v_i ,\sum_{j = 2}^lc_j \lambda_j v_j  \right\rangle_\pi\\
    & = \sum_{i = 2}^l \sum_{j = 2}^lc_i \lambda_i c_j \lambda_j \langle v_i,v_j  \rangle_\pi\\
    & = \sum_{i = 2}^lc_i^2\lambda_i^2\\
    & \leq \lambda_2^2\sum_{i = 2}^lc_i^2  = \lambda_2^2||f||^2_{2,\pi}.
\end{align*}
Since $\lambda_{2} \in (0, 1)$, taking square roots gives 
\begin{equation}
||f \widehat{P}_{\varepsilon}||_{2,\pi} \leq \lambda_2||f||_{2,\pi} = (1+ \beta_{2}\varepsilon)||f||_{2,\pi}.
\label{eq:inhomo_eight}
\end{equation}
Let $\pi^{-1}$ denote the vector with entries $\frac{1}{\pi(m)}$. For two vectors $x$ and $x'$ with entries $x(m)$ and $x(m')$, let $x \circ x'$ denote the Hadamard product, with entries $x(m)x(m')$. Note that we have $\langle \pi^{-1} \circ e,\mathbf{1}\rangle_{\pi} = 0$ and that, therefore, \eqref{eq:inhomo_eight} holds for $f =\pi^{-1} \circ e$. Let $D_x$ denote the diagonal matrix with diagonal entries $x(m)$, and let $||\cdot||_{op}$ denote the matrix operator norm corresponding to the $||\cdot||_{2,\pi}$ vector norm. For a matrix $A$, we have $||(x \circ x')A||_{2,\pi} \leq ||D_x||_{op}||x'A||_{2,\pi}$. We have
\begin{align*}
||e \widehat{P}_{\varepsilon} ||_{2,\pi} & = ||(\pi \circ \pi^{-1} \circ e) \widehat{P}_{\varepsilon} ||_{2,\pi}\\
& \leq || \pi^{-1} \circ e) \widehat{P}_{\varepsilon} ||_{2, \pi} ||D_\pi||_{op}\\
& \leq (1+\beta_{2} \varepsilon)
||\pi^{-1} \circ e ||_{2,\pi} 
||D_\pi||_{op} \\
& \leq (1+\beta_{2}\varepsilon)
|| e ||_{2,\pi}
||D_{\pi^{-1}}||_{op}
||D_\pi||_{op} \\
& \leq (1+\beta_{2}\varepsilon)
|| e ||_{2,\pi}
||D_{\pi^{-1}}D_\pi||_{op} \\
& = (1+\beta_{2}\varepsilon)
|| e ||_{2,\pi}.
\end{align*}
 Finally, let $\beta = -\beta_{2} > 0$. This concludes the proof.
\end{proof}

\subsection{Proof of Theorem 3.1} \label{sup: Proof of Theorem homo}

We prove the claims (i)-(iv) in order.

\begin{enumerate}[label = (\roman*)]
\item Since $\pi(m) > 0$ for all $m \in \mathcal{M}$, the BD process is irreducible. Adding the fact that $\mathcal{M}$ is finite-space, we conclude that $\{ M_{\varepsilon}^{(s)} \}_{s \geq 1}$ is irreducible and aperiodic. Hence, the stationary distribution $\pi_{\varepsilon}$ is unique and equal to the limiting distribution.

\item 


By Lemma \ref{lem: P in terms of I and Q}, we have
\begin{equation} \label{eq:inhomo_step0}
P_{\varepsilon} = I + \varepsilon Q + E_\varepsilon,
\end{equation}
where $E_{\varepsilon} = \mathcal{O}(\varepsilon^{2})$. Next, let $e_{\varepsilon} := \pi_{\varepsilon} -  \pi$ denote the error vector. Using $\pi Q = 0$, and $\pi_{\varepsilon} = \pi_{\varepsilon}  P_{\varepsilon}$, and Equation \eqref{eq:homo_step1}, we have
    \begin{align}
        e_{\varepsilon} & = \pi_{\varepsilon}-\pi \nonumber\\
        & = \pi_{\varepsilon} P_{\varepsilon} - \pi \nonumber\\
        & = (\pi + e_\varepsilon)(I  + \varepsilon Q + E_{\varepsilon}) - \pi \nonumber\\
        & = \pi + \varepsilon \pi Q + \pi E_{\varepsilon} + e_{\varepsilon} + \varepsilon e_{\varepsilon} Q + e_{\varepsilon} E_{\varepsilon} - \pi \nonumber\\
        & = e_{\varepsilon}(I + \varepsilon Q) + (\pi+e_{\varepsilon}) E_{\varepsilon}.
        \label{eq:home_four}
    \end{align}
    Applying the triangular inequality on \eqref{eq:home_four}, we have 
\begin{equation}
||e_{\varepsilon}||_{2,\pi} \leq ||e_\varepsilon(I + \varepsilon Q)||_{2,\pi} + ||(\pi+e_\varepsilon) E_\varepsilon ||_{2,\pi}. \label{eq:homo_five}
\end{equation}

Since the elements of $e_\varepsilon$ sum to $0$, by Lemma \ref{lem:inhomo}, there exists a $\beta > 0$ and an $\varepsilon^{\ast} > 0$ such that for all $\varepsilon < \varepsilon^*$, we have
\begin{equation}||e_{\varepsilon}(I + \varepsilon Q)||_{2,\pi} \leq (1-\varepsilon \beta)||e_{\varepsilon}||_{2,\pi}.
\label{eq:homo_seven}
\end{equation}
For the second term in \eqref{eq:homo_five}, since $E_s$ is a matrix with entries of order $\mathcal{O}(\varepsilon^2)$, we have for some constant $C>0$ that 
\begin{equation}
    ||(\pi + e_\varepsilon)E_s||_{2,\pi} \leq  C\varepsilon^2.\label{eq:homo_six}
\end{equation}
Substituting \eqref{eq:homo_seven} and \eqref{eq:homo_six} into \eqref{eq:homo_five}, we have
\begin{equation} \label{eq:homo_bound}
    ||e_{\varepsilon}||_{2,\pi} \leq  (1 - \varepsilon \beta)||e_{\varepsilon}||_{2,\pi} + C\varepsilon^2,
\end{equation}
for some $\beta > 0, C > 0$ and all $\varepsilon < \varepsilon^{\ast}$. Rewriting \eqref{eq:homo_bound}, we get
\begin{equation} \label{eq:homo_bound_rewritten}
    ||e_{\varepsilon}||_{2,\pi} \leq  \frac{C\varepsilon}{\beta}.
\end{equation}
Since any two vector norms on a finite space are equivalent, there exists, for any vector norm $|| \cdot ||$, a constant $A > 0$, such that
\begin{equation}
||e_\varepsilon|| \leq A ||e_{\varepsilon}||_{2,\pi} \leq \frac{AC}{\beta}\varepsilon.
\label{eq:bound}
\end{equation}
Hence, $||e_\varepsilon|| = O(\varepsilon)$.

\item For the CTMC $(X(t))_{t \geq 0}$, the waiting time $W_{m}$ in state $m$ follows an exponential distribution with rate $\lambda(m) := \sum_{i=1}^{k} q_i(m)$. It therefore has the following cumulative distribution function:
\begin{equation*}
    \Pr(W(m) \leq u) = 1 - e^{-\lambda(m) u}.
\end{equation*}
For the DTMC in Algorithm 1 with fixed $\varepsilon_{s} = \varepsilon$, first note that the rejection (self-transition) probability in state $m$ is
\begin{equation*}
    r_{\varepsilon}(m) := P_{\varepsilon}(m,  m) = \prod_{i=1}^{k}(1-q_i(m)\varepsilon).
\end{equation*}
Hence, for the waiting times, we have
\[
\Pr(\varepsilon W_{m}^{\varepsilon} \leq u) = \Pr(W_{m}^{\varepsilon} \leq u/\varepsilon) = 1 - r_{\varepsilon}(m)^{\lfloor u/\varepsilon \rfloor}, \quad u \geq 0.
\]
In the limit, this cumulative distribution function (cdf) converges to the cdf of the CTMC. That is,
\begin{equation}
\label{eq:cdf_is_cdf}
    \lim_{\varepsilon \to 0} \Pr(\varepsilon W_{m}^{\varepsilon} \leq u) = 1 - e^{-\lambda(m) u}.
\end{equation}
To see that \eqref{eq:cdf_is_cdf} holds, we write
\begin{align*}
 \lim_{\varepsilon \to 0} \left\lfloor \frac{u}{\varepsilon} \right\rfloor r_{\varepsilon}(m) &= \lim_{\varepsilon \to 0} \frac{u}{\varepsilon} \ln r_{\varepsilon}(m) + \lim_{\varepsilon \to 0} \mathcal{O}(\varepsilon) \ln r_{\varepsilon}(m) \\
 &= \lim_{\varepsilon \to 0} \frac{ u \ln r_{\varepsilon}(m) }{\varepsilon} + 0 \\
 & = \frac{\lim_{\varepsilon \to 0}  u \sum_{i = 1}^{k} \frac{-q_{i}(m)}{1 - q_{i}(m) \varepsilon}}{1} \\
 & = - u \sum_{i = 1}^{k} q_{i}(m) = -u \lambda(m),
\end{align*}
where the first line is due to $u/\varepsilon - \lfloor u/\varepsilon \rfloor = \mathcal{O}(\varepsilon)$, the second line is because the second term converges to zero, and the third line is due to the L'Hopital rule. Hence, 
\begin{align*}
\lim_{\varepsilon \to 0} \mathbb{P}(\varepsilon W_{m}^{\varepsilon} \leq u) & = 1 - \lim_{\varepsilon \to 0} r_{\varepsilon}(m)^{\lfloor u/\varepsilon \rfloor} \\
& = 1 - \lim_{\varepsilon \to 0} \exp\left\{  \left\lfloor \frac{u}{\varepsilon} \right\rfloor \ln r_{\varepsilon}(m) \right\} \\
& = 1 - \exp\left\{ \lim_{\varepsilon \to 0} \left\lfloor \frac{u}{\varepsilon} \right\rfloor \ln r_{\varepsilon}(m) \right\} \\
& = 1 - e^{- u \lambda(m)}.
\end{align*}
Therefore, as $\varepsilon \rightarrow 0$, the (scaled) waiting time $\varepsilon W_{m}^{\varepsilon}$ converges in distribution to an exponential random variable with rate $\lambda(m)$, which is exactly the distribution of the waiting time of the birth-death CTMC. 

\item For the CTMC, the chain moves to a new state $m'$ with a probability
\begin{equation*}
\widetilde{P}(m, m') = \begin{cases}
       \frac{q_i(m)}{\lambda(m)}, & \text{if } m' = m^i \text{ for some } i \in \{1,\ldots,k\};\\
       0,  & \text{otherwise}.
   \end{cases}
\end{equation*}
Assume now that our DTMC is in state $m$ after $s$ iterations. That is, $m^{(s)} = m$. Given that our chain moves to new state $m^{(s+1)} \neq m$, the probability that this new state is $m'$ is 
\begin{align*}
\widetilde{P}_{\varepsilon}(m, m') & := \Pr[m^{(s+1)} = m'  \mid  m^{(s)} = m, m^{(s+1)} \not = m^{(s)}] \\
& = \frac{\Pr(m^{(s+1)} = m'  \mid  m^{(s)}=m)}{\Pr(m^{(s+1)} \neq m  \mid  m^{(s)}=m]} \\
& = \frac{P_{\varepsilon}(m, m')}{1-r_{\varepsilon}(m)}.
\end{align*}
We  consider two cases: $m' \in \{m^1,\ldots,m^k\}$ and $m' \not \in \{m^1,\ldots,m^k\}$. For the first case, where $m = m^{i}$ for some $i \in \{1, \ldots, k\}$, we have
\begin{align*}
\widetilde{P}_{\varepsilon}(m, m^{i}) &= \frac{\varepsilon q_{i}(m) \prod_{j \neq i} (1 - \varepsilon q_{j}(m))}{1 -\prod_{j = 1}^{k} (1 - \varepsilon q_{j}(m))} \\
&= \frac{\varepsilon q_{i}(m) }{1 -\prod_{j = 1}^{k} (1 - \varepsilon q_{j}(m))} \prod_{j \neq i} (1 - \varepsilon q_{j}(m)).
\end{align*}
Note that $\lim_{\varepsilon \to 0} \prod_{j \not= i} (1 - \varepsilon q_{j}(m)) = 1$. Moreover, by the L'Hopital rule, the other factor converges to
\begin{align*}
\lim_{\varepsilon \to 0} \frac{\varepsilon q_{i}(m) }{1 -\prod_{j = 1}^{k} (1 - \varepsilon q_{j}(m))} = \frac{q_{i}(m)}{\lim_{\varepsilon \to 0}\sum_{j = 1}^{k} q_{j}(m) \prod_{l \neq j} (1 - \varepsilon q_{l}(m))} = \frac{q_{i}(m)}{\sum_{j = 1}^{k} q_{j}(m)}.
\end{align*}
Hence, 
\[
 \lim_{\varepsilon \to 0} \widetilde{P}_{\varepsilon}(m, m^i) = \lim_{\varepsilon \to 0} \frac{P_{\varepsilon}(m,m^{i})}{1 -r_{\varepsilon}(m)} = \frac{q_{i}(m)}{\sum_{j = 1}^{k} q_{j}(m)} = \frac{q_{i}(m)}{\lambda(m)} = \widetilde{P}(m, m^i).
\]
The proof for the second case, where $m' \not \in \{m^1,\ldots,m^k\}$, follows naturally. Note that the probability that the chain moves from a state $m$ to a neighboring state converges as 
\[ 
\lim_{\varepsilon \to 0}
\sum_{i=1}^{k} \widetilde{P}_{\varepsilon}(m, m^i) = \sum_{i=1}^{k} \frac{q_i(m)}{\lambda(m)} = 1.
\]
Hence, when $m' \not\in \{m^1,\ldots,m^k\}$, $m'$ is not a neighbor of $m$ and $\widetilde{P}_{\varepsilon}(m, m') = 0 = \widetilde{P}(m, m')$.
\end{enumerate}

\subsection{Proof of Theorem 3.2} \label{sup: Proof of Theorem inhomo} 
As before, let $l := 2^k$. For $s=1,2,\ldots$, let $P_s$ denote the transition probability matrix at iterations $s$ of our DTMC. That is, $P_s$ is the $l \times l$ matrix with elements $P_{\varepsilon_s}(m,m')$ given by Equation 8 of the main document. Let $\mu_0$ be the initial distribution of our DTMC and let $\mu_s$ denote the probability distribution at iteration $s$. That is,
\begin{equation}
\mu_s := \mu_0 P_1 P_2 \ldots P_s \quad s \geq 0.
\end{equation}

Our reversible BD process has invariant distribution $\pi$ and an $l \times l$ rate matrix $Q$. This proof will show that the inhomogeneous DTMC has the same limiting distribution as the BD process, as long as $\lim_{s \to \infty} \varepsilon_s = 0$ and $\sum_{s=1}^\infty\varepsilon_s = \infty$. That is, we will show that 
\begin{equation}
\lim_{s \to \infty}||\mu_s-\pi|| = 0,
\end{equation}
where $||\cdot||$ can be any vector norm. Note that we denote both $\mu_s$ and $\pi$ as $l$-dimensional vectors with elements $\mu_s(m)$ and $\pi(m)$ respectively. Lastly, define the $l_2(\pi)$-norm of a vector $x = (x(1),\ldots,x(l))$ as $\| x \|_{2,\pi} = \left( \sum_{i} x(i)^2\pi(i) \right)^{1/2}$.  Note that this norm is well-defined since $\pi(m) > 0$ for all $m \in \mathcal{M}$, which is ensured by the irreducibility of the chain $Q$ and the finiteness of $\mathcal{M}$. 

\noindent We will now prove Theorem 3.2

\begin{enumerate}
    \item \textbf{Writing the DTMC as a function of the BD process:} First, by Lemma \ref{lem: P in terms of I and Q}, we can write
    \begin{equation} \label{eq:inhomo_step1}
    P_{s} = I + \varepsilon_s Q + E_s,
    \end{equation}
    where $E_s$ is an error matrix with entries of order $\mathcal{O}(\varepsilon_s^2)$. 
%
    \item \textbf{Evolution of the error term:} Let $e_s := \mu_s -  \pi$ denote the error vector at iteration $s$. Using $\pi Q = 0$, $\mu_{s+1} = \mu_s P_{s}$, and Equation \eqref{eq:inhomo_step1}, we have
    \begin{align}
        e_{s+1} & = \mu_{s+1}-\pi \nonumber\\
        & = \mu_sP_{s} - \pi \nonumber\\
        & = (\pi + e_s)(I  + \varepsilon_sQ + E_s) - \pi \nonumber\\
        & = \pi + \varepsilon_s\pi Q + \pi E_s + e_s + \varepsilon_s e_s Q + e_s E_s - \pi \nonumber\\
        & = e_s(I + \varepsilon_sQ) + (\pi+e_s) E_s.
        \label{eq:inhome_four}
    \end{align}
\item \textbf{Bounding of the error term:} Applying the triangular inequality on \eqref{eq:inhome_four},  we have 
\begin{equation}
||e_{s+1}||_{2,\pi} \leq ||e_s(I + \varepsilon_s Q)||_{2,\pi} + ||(\pi+e_s) E_s ||_{2,\pi}. \label{eq:inhomo_five}
\end{equation}
Since $e_s$ sums to $0$, by Lemma \ref{lem:inhomo}, there exists an $s^* >0$, such that for all $s \geq s^*$, there exists a $\beta > 0$ such that
\begin{equation}
    || e_s (I + \varepsilon Q)||_{2,\pi} \leq (1-\varepsilon \beta)||e_s ||_{2,\pi}.
\end{equation}
The existence of $s^{\ast}$ is guaranteed by the condition that $\varepsilon_{s} \to 0$. For the second term in \eqref{eq:inhomo_five}, since $E_s$ is a matrix with entries of order $\mathcal{O}(\varepsilon_s^2)$, we have that 
\begin{equation}
    ||(\pi + e_s)E_s||_{2,\pi} \leq  C\varepsilon_s^2,\label{eq:inhomo_six}
\end{equation}
for some constant $C>0$. Hence, there exists an $s^* >0$ and $\beta > 0$, such that for all $s \geq s^*$, we have
\begin{equation} 
    ||e_{s+1}||_{2,\pi} \leq  (1 - \varepsilon_s \beta)||e_{s}||_{2,\pi} + C\varepsilon_s^2
\end{equation}
for some $\beta, C > 0$ and all $s > s^{\ast}$.

\item \textbf{Unwinding recurrence:} We now analyze the recurrence relation 
\begin{equation} \label{eq:step4_1}
r_{s+1} \leq (1 - \beta \varepsilon_s) r_s + C \varepsilon_s^2,
\end{equation} where $r_s := ||e_s||_{2,\pi}$. We will show that this recurrence relation, together with the decay and divergence in the sequence $\{ \varepsilon_s\}_{s \geq 1}$, implies that $\lim_{s\to \infty} r_s = 0$.

    Let $\delta \in (0, r_{0})$ be an arbitrary number. Since $\varepsilon_s \to 0$, we know that there exists an $s' > s^{\ast}$, such that for all $s \geq s'$, we have $\varepsilon_s < \frac{\delta \beta}{C}$, and therefore
    \begin{equation}
    C\varepsilon_s^2 < \delta \beta \varepsilon_s \quad \text{ for all } s \geq s'.
    \label{eq:step4_2}
    \end{equation}
    We can combine Equation \eqref{eq:step4_1} and \eqref{eq:step4_2} to obtain
    \begin{equation}
        r_{s+1} < (1 - \beta \varepsilon_s) r_s + \delta \beta \varepsilon_s \quad \text{ for all } s \geq s'.
        \label{eq:step4_3}
    \end{equation}
    Now, imagine another sequence $\{z_s\}_{s \geq s'}$ that has $z_{s'} = r_{s'}$, and 
    \begin{equation}
        z_{s+1} = (1 - \beta \varepsilon_s) z_s + \delta \beta \varepsilon_s \quad \text{ for all } s \geq s'.
    \end{equation}
    We therefore have $r_s < z_s$ for all $s > s'$. Moreover, note that we can write 
    \begin{equation}
        z_{s+1} - \delta = (1 - \beta \varepsilon_s)(z_s - \delta)\text{ for all } s \geq s'.
    \end{equation}
    Therefore, we can use the recursive product to write
    \begin{equation}
        z_{s} - \delta = (z_{s'} - \delta) \prod_{i=s'}^{s-1}(1 - \beta \varepsilon_s)\text{ for all } s \geq s'.
    \end{equation}
    Since $\beta\varepsilon_{s} \in (0, 1)$, We then obtain a bound on $z_s$ by
    \begin{align}
    z_s & = \delta + (z_{s'} - \delta)\prod_{i=s'}^{s-1}(1 - \beta \varepsilon_s)\\
    & \leq \delta + (z_{s'} - \delta)\prod_{i=s'}^{s-1}e^{-\beta \varepsilon_s}\\
    & = \delta + (z_{s'} - \delta)\exp\left(-\beta \sum_{i=s'}^{s-1}\varepsilon_s\right).
    \end{align}
    Since $\sum_{i=s'}^{\infty}\varepsilon_s = \infty$, we have 
    \begin{equation}
    \lim_{s\to \infty}||\mu_s - \pi||_{2,\pi} = \lim_{s\to \infty}||e_s||_{2,\pi} = 
     \lim_{s \to \infty}  r_s < \lim_{s \to \infty} z_s \leq \delta,
    \end{equation}
    but since $\delta >0$ can be arbitrarily small, we obtain $\lim_{s\to \infty}||\mu_s - \pi||_{2,\pi} = 0$.   
    \item \textbf{Concluding}: Since any two vector norms on a finite vector space are equivalent, we have that $\lim_{s\to \infty}||\mu_s - \pi|| =  0$ for any vector norm $||\cdot||$.

\end{enumerate}

\subsection{Proof of Lemma 3.3} \label{sup:proof_fact_bd_process}
We assume that the posterior is factorizable as in Equation 10. Consider the BD process with rates 
\begin{equation*}
q_i(m) = \min\left\{1,\frac{p(m^{i} \mid y)}{p(m \mid y)}\right\}
\end{equation*}
and invariant distribution $\pi(m) = p(m \mid y)$. Consider also the homogeneous DTMC with transition probabilities
\begin{equation*} 
    P_{\varepsilon}(m, m') = \prod_{i \in H_{m, m'}}q_i(m)\varepsilon \prod_{i \not \in H_{m, m'}}(1-q_i(m)\varepsilon),
\end{equation*}
for some $\varepsilon \in (0,1)$. Here $H_{m, m'} =\{i_1,\ldots, i_r\} \subseteq \{1,\ldots,k\}$ is the set of elements in which models $m$ and $m'$ differ. We will show that, for any $\varepsilon \in (0,1)$, the DTMC will also have $\pi(m)$ as its invariant distribution.
\begin{enumerate}
    \item \textbf{Rate equivalence:} In this step we show that 
    \begin{equation}
        q_i(m) = q_i(m') \quad \text{ for all } m,m' \quad \text {with } m_i = m'_i.
        \label{eq:lemma_3.3_one}
    \end{equation} 
    We have
    \begin{align*}
        \frac{p(m^i \mid y)}{p(m \mid y)} & = \frac{\prod_{i=1}^kp(m^i_i \mid y)}{\prod_{i=1}^kp(m_i \mid y)} \\
        & = \frac{p(m^i_i \mid y)}{p(m_i \mid y)} \\
        & = \frac{p(1-m_i \mid y)}{p(m_i \mid y)}.
    \end{align*}
    Hence, the rate $q_i(m)$ only depends on $m_i$ and equation \eqref{eq:lemma_3.3_one} holds.

    \item \textbf{Detailed balance:} Let $m,m'$ be arbitrary models. For all $i \not \in H$, we have $m_i = m'_i$, and hence $q_i(m) = q_i(m')$, and therefore
    \begin{equation}
        \prod_{i \not \in H}(1-q_i(m)\varepsilon) = \prod_{i \not \in H}(1-q_i(m')\varepsilon).
        \label{eq:lemma_3.3_two}
    \end{equation}
    For all $i \in H$, we have $m_i = m'^i_i$, and hence $q_i(m) = q_i(m'^i)$. Therefore, we have 
    \begin{equation}
    \label{eq:lemma_3.3_three}
    \pi(m) \prod_{i \in H}q_i(m) \varepsilon = \pi(m') \prod_{i \in H}q_i(m')\varepsilon,
    \end{equation}
    because
\begin{align*}
    \pi(m) \prod_{i \in H}q_i(m)\varepsilon & = \varepsilon^r\pi(m) q_{i_1}(m)q_{i_2}(m)\ldots q_{i_r}(m) \\
    & = \varepsilon^r\pi(m) q_{i_1}(m)q_{i_2}(m^{i_1})\ldots q_{i_r}(m^{i_1i_2\ldots i_{r-1}}) \\
    & = \varepsilon^r q_{i_1}(m^{i_1})q_{i_2}(m^{i_1i_2})\ldots q_{i_r}(m^{i_1i_2\ldots i_r})\pi(m^{i_1i_2\ldots i_r})\\
    & = \varepsilon^r q_{i_1}(m')q_{i_2}(m')\ldots q_{i_r}(m')\pi(m')\\
    & = \pi(m') \prod_{i \in H}q_i(m')\varepsilon
\end{align*}
Here, in the 3rd equality, we use that $\pi(m)q_i(m) = \pi(m^i)q_i(m^i)$ for all $i=1,\cdots,k$ and all $m \in \mathcal{M}$. 
Now, combining \eqref{eq:lemma_3.3_two} and \eqref{eq:lemma_3.3_three} gives the required detailed balance:
\begin{equation*}
    \pi(m) P_{\varepsilon}(m, m') = \pi(m') P_{\varepsilon}(m', m).
\end{equation*}
We conclude that $\pi(m)$ is also the invariant distribution for the DTMC.
\end{enumerate}

\section{Applications}
In Section 4.1 of the main article, we apply our work to undirected Gaussian graphical models. Here, we also apply our approach to Bayesian variable selection and Ising models. For each application, the path to construct our Multiple Jump algorithm is the same: we formulate the model posterior and specify the rates $q_i(m) := Q(m,m^i)$ for the BD process in Lemma 2.1 of the main article. Using these rates in Algorithm 1 will give the desired algorithm.

\subsection{Bayesian variable selection}
\label{sec:app_BVS}
This subsection applies our work to Bayesian variable selection. Variable selection -- also referred to as feature selection -- aims to find a subset of variables that best explains an outcome of interest. We treat here the most fundamental form of variable selection: the multiple linear regression. Literature has suggested many approaches for this problem. Here, we will treat the approach using the g-prior \citep{Fernandez2004}. We propose a BD process that can be extended to a superior DTMC.\\

\paragraph{Preliminaries:} We assume that a continuous variable $Y$ is modeled via a linear combination of $k$ explanatory variables $X_1,\ldots,X_k$. We observe $n$ observations of all our variables. Let $y = (y_1,\ldots,y_n)$ denote the $n$ observations of variable $Y$, and let the $n \times k$ matrix $x$ contain the corresponding $n$ observations from the variables $X_1,\ldots,X_k$. Its elements $x_{ij}$ denote the $i$th observation of variable $X_j$, $i=1,\ldots,n$ and $j=1,\ldots,k$. We then have
\begin{equation}
    y_i = x_{i1}\beta_1 + \cdots + x_{ik}\beta_k + \varepsilon_i, \quad i=1,\ldots,n,
\end{equation}
where $\varepsilon_i \sim \mathcal{N}(0,1)$ and $\beta = (\beta_1,\ldots,\beta_k)$ is the vector of unknown regression coefficients. We define a ``model'' as a binary vector $\gamma = (\gamma_1,\ldots,\gamma_k )$, where $\gamma_i$ is equal to $1$ if variable $X_i$ is deemed relevant and $0$ otherwise. The regression coefficients $\beta_{k}$ are the ``model-associated parameters''. Bayesian variable selection sets out to find the posterior distribution of this model, given by 
\begin{equation}
p(\gamma  \mid  y,x) = \frac{p(y,x  \mid  \gamma ) p(\gamma)}{p(y,x)}.
\label{eq:posterior_BVS}
\end{equation}
For a review of Bayesian variable selection, see \citet{Tadesse2021}.

\paragraph{BD processes:} Using the g-prior \citep{Fernandez2004} for the regression coefficients, the Bayes factor simplifies to 
\begin{equation}
\frac{p(y,x  \mid  \gamma^i)}{p(y,x \mid  \gamma)} = (1+g)^{(|\gamma|-|\gamma^i|)1/2}\frac{(1+g(1-R^2_{\gamma}))^{\frac{n-1}{2}}  }{(1+g(1-R^2_{\gamma^i}))^{\frac{n-1}{2}}},
\label{eq:BVS_1}
\end{equation}
where $\gamma^i$ equals the vector $\gamma$ except in element $\gamma_i$, $g \in \mathbb{R}^+$ is a user-specified constant, $R^2_{\gamma}$ is the coefficient of determination for model $\gamma$, and $|\gamma|$ denotes the number of non-zero entries in a model $\gamma$. Since $\gamma^i$ and $\gamma$ differ in only one element, we have $(|\gamma|-|\gamma^i|) \in \{-1,1\}$. With this exact evaluation of the Bayes factor, one can design a BD algorithm as in Lemma 2.1 in the main article with rates
\begin{equation}
Q(\gamma, \gamma') = \begin{cases} q_i(\gamma), & \text{if }  \gamma' = \gamma^i \text{ for some $i \in \{1,\ldots,k\}$} ; \\
0, & \text{otherwise},
\end{cases}
\end{equation}
where
\begin{equation}
q_i(\gamma) := \min \left\{1, \frac{p(y,x  \mid  \gamma^i)}{p(y,x \mid  \gamma)}\frac{p(\gamma^i)}{p(\gamma)} \right\}.
\label{eq:rates_gprior}
\end{equation}

This BD process will add/delete exactly one variable in each iteration. This becomes problematic for large model spaces. 

\paragraph{Our algorithm:} We can extend the BD process to a DTMC with transition probabilities (8) and use Algorithm 1 to generate samples $\gamma^{(1)},\ldots,\gamma^{(S)}$. These samples will have the same limiting distribution as the corresponding BD process as long as the sequence $\{ \varepsilon_{s}\}_{s \geq 1}$ satisfies (9), see Theorem 3.2. However, the DTMC is able to add/delete multiple edges in a single iteration, and therefore has the potential to lead to a significant reduction in computation time. 

\subsection{Ising model}
\label{sec:app_Ising}
This subsection treats the Ising model. We introduce the Ising model and propose a BD process, which we extend to a superior DTMC.

\paragraph{Preliminaries:}
In the Ising model, the data $y$ is an $n \times p$ matrix, containing $n$ observations of the variables $Y_1,\ldots,Y_p$. Each observation $y = (y_1,\ldots,y_p)$ is a $p$-dimensional binary vector and follows the following distribution:
\begin{equation*}
p(y  \mid  \mu, \Sigma)
= \frac{1}{Z(\mu, \Sigma)} 
\exp\!\left(
\sum_{e=(i,j): i <j} \sigma_e\, y_i y_j
+ \sum_{i} \mu_i\, y_i
\right), \quad i,j = 1,\ldots,p.
\end{equation*}
Here, the vector $\mu = (\mu_1,\ldots,\mu_p)$ contains the main effects, the symmetric matrix $\Sigma$ -- with elements $\sigma_{e}$ for every edge $e = (i,j)$ -- contains the associations between each pair of variables, and $Z(\mu, \Sigma)$ denotes the normalizing constant. The matrix $\Sigma$ is called the interaction matrix and encodes the conditional independence between the variables due to the following relationship:
\begin{equation*}
    \sigma_{e} = 0 \text{ for some edge } e = (i,j)\iff Y_i \text{ and } Y_j \text{ are conditionally independent}. 
\end{equation*}
Similar to GGMs (Subsection 4.1), we can depict the conditional dependence relationships in a graph $G$. Bayesian inference in Ising models aims to uncover the posterior distribution of this graph given by
\begin{equation}
p(G \mid y) = \frac{p(y \mid G)p(G)}{p(y)}.
\end{equation}
As in GGMs, we can encode every graph $G$ as a binary vector of length $k=p(p-1)/2$. The graphs $G$ are therefore the binary ``models,'' and the space of all $p \times p$ graphs is the binary model space. The interaction matrix $\Sigma$ is the model-associated parameter. 

\paragraph{BD process:} 
Using the marginal pseudo-likelihood \citep{besag1975}, we can factorize the marginal likelihood $p(y \mid G)$ nodewise and obtain
\begin{equation}
    p(y \mid G) \approx \tilde{p}(y \mid G) := \prod_{i=1}^pp_i(G),
\end{equation}
where $p_i(G)$ is the marginal likelihood of a logistic regression of variable $Y_i$ on all variables $Y_j$ that are neighbors of node $i$ in the graph $G$. Following \cite{Barber2015}, we can approximate each $p_i(G)$ using the (extended) Bayesian Information Criterion. Let $G$ and $G^e$ be two graphs that differ only in the edge $e = (i,j)$. Their Bayes factor can now be approximated by 
\begin{equation}
\frac{p(y \mid G^e)}{p(y \mid G)} \approx \frac{\tilde{p}(y \mid G^e)}{\tilde{p}(y \mid G)} = \frac{p_i(G^e)p_j(G^e)}{p_i(G)p_j(G)}.
\end{equation}
We can now design a BD process with rates 
\begin{equation}
Q(G, G') = \begin{cases} q_e(G), & \text{if } G' = G^e \text{ for some edge $e$};\\
0, & \text{otherwise,}
\end{cases}
\end{equation}
where 
\begin{equation}
q_e(G) := \min \left\{1,\frac{p_i(G^e)p_j(G^e)}{p_i(G)p_j(G)} \frac{p(G^e)}{p(G)} \right\}.
\label{eq:rates_Ising}
\end{equation}
Such a BD process would add/delete a single edge per iteration. It would therefore be problematic to apply it to large-scale problems.

\paragraph{Our algorithm:} Our work allows us to extend this BD process to a DTMC with transition probabilities (8), where the rates $q_e(G)$ are as in \eqref{eq:rates_Ising}. One can use Algorithm 1 to sample from this DTMC. If the sequence $\{ \varepsilon_{s}\}_{s \geq 1}$ satisfies (9), the resulting DTMC samples have the same limiting distribution as the BD process (see Theorem 3.2). Contrary to the BD process, however, the DTMC is able to add or delete multiple edges in a single iteration. It therefore has the potential to deliver a significant reduction in computational cost.

\section{Extended simulation results} \label{sup:appendix_simulations}

This section includes a detailed description of all instances, settings, metrics, and results of the simulation study of Section 5.1 of the main document.

\paragraph{Instances:} We only consider instances with $p=1000$ variables. We consider instances with a small ($n=400$) and large ($n=1000)$ amount of observations. We consider sparse ($0.2\%$ edge density) and dense $(1\%)$ graphs. All instances are listed in Table \ref{tab:sim_settings}.

\paragraph{Simulation settings:} For an instance with $p$ variables, $n$ observations and edge density $\alpha \in (0,1)$, we generate a true graph $G^*$ at random by including each edge with a probability of $\alpha$. We then sample the true precision matrix $K^*$ from the G-Wishart distribution with $D = I$ as the scale parameter and $b=3$ as the shape parameter. We then sample the data $y$ containing $n$ observations from the $p$-variate normal distribution with covariance matrix $K^{*-1}$. We repeat this process eight times. We run our algorithms on $8$ cores with $8$ GiB memory per core. Both the BD-MPL and the MJ-MCMC algorithms are implemented in C++ and called in R. We use a prior density of $0.5\%$ and initialize our algorithm at the empty graph. Depending on the instances, the number of MCMC iterations to run our algorithm ranges between $100$ and one million. The MCMC iterations for each instance are listed in Table \ref{tab:sim_settings}.

\paragraph{Maximum jump size:} We observed that in some cases, when starting from an empty graph, and only in the first five iterations of the MCMC chain, our algorithm suggests very large jumps. These jumps are so large that in those first iterations, our algorithm outputs samples with an edge density as high as $50\%$. Later in the chain, the algorithm corrects this ``overshooting" and outputs MCMC samples with a normal edge density (i.e. $~1\%$) again. The initial overshooting, however, requires lots of computational power. That is why, during those first iterations, we restrict the jump size to $r p(p-1)/2$ edges per iteration, where we refer to $r\in (0,1)$ as the maximum jump size. $r$ denotes the percentage of edges that is allowed to flip in a single iteration. Depending on the instance, $r$ ranges from $0.0025$ to $1$. See Table \ref{tab:sim_settings} for the instance-specific values of the maximum jump size $r$.

\paragraph{Metrics:} We use four different metrics to assess the accuracy of the estimated edge-inclusion probabilities: AUC-PR, AUC-ROC, $p^+$, $p^-$. The AUC-PR represents the area under the precision-recall curve. The AUC-ROC represents the area under the receiver operating characteristic curve. The $p^+$ ($p^-$) values represent the average posterior edge inclusion probability for edges that are (not) in the true graph. The AUC-PR, AUC-ROC, and $p^+$ values range from $0$ (worse) to $1$ (best). The $p^-$ value ranges from $0$ (best) to $1$ (worse).

\begin{table}[]
\begin{tabular}{llllllll}
\multicolumn{1}{c}{\textbf{p}} & \multicolumn{1}{c}{\textbf{density}} & \multicolumn{1}{c}{\textbf{n}} & \multicolumn{1}{c}{\textbf{algorithm}} & \multicolumn{1}{c}{
$\varepsilon_1$} & \multicolumn{1}{c}{\textbf{decay} in $\varepsilon_s$} & \multicolumn{1}{c}{\textbf{max jump}} & \multicolumn{1}{c}{\textbf{MCMC Iterations}} \\
\hline
$1000$                         & $0.2\%$                               & $400$                          & MJ-MCMC                                & $0.001$                              & none                               & $1$                                   & $40000$                                      \\
                               &                                      &                                &                                        & $0.01$                               & none                               & $1$                                   & $4000$                                       \\
                               &                                      &                                &                                        & $0.05$                               & none                               & $1$                                   & $750$                                        \\
                               &                                      &                                &                                        & $0.3$                                & none                               & $1$                                   & $200$                                        \\
                               &                                      &                                &                                        & $0.7$                                & none                               & $0.0025$                              & $100$                                        \\
                               &                                      &                                &                                        & $0.95$                               & none                               & $0.0025$                              & $100$                                        \\
                               &                                      &                                &                                        & $0.3$                                & slow                               & $1$                                   & $500$                                        \\
                               &                                      &                                &                                        & $0.3$                                & fast                               & $1$                                   & $40000$                                      \\
                               &                                      &                                & \multicolumn{1}{c}{MPL-BD}             & \textit{NA}                          & \textit{NA}                        & \textit{NA}                           & $300000$                                     \\
\hline
$1000$                         & $0.2\%$                               & $1000$                         & MJ-MCMC                                & $0.001$                              & none                               & $1$                                   & $40000$                                      \\
                               &                                      &                                &                                        & $0.01$                               & none                               & $1$                                   & $4000$                                       \\
                               &                                      &                                &                                        & $0.05$                               & none                               & $1$                                   & $750$                                        \\
                               &                                      &                                &                                        & $0.3$                                & none                               & $0.0025$                              & $1000$                                       \\
                               &                                      &                                &                                        & $0.7$                                & none                               & $0.0025$                              & $100$                                        \\
                               &                                      &                                &                                        & $0.95$                               & none                               & $0.0025$                              & $100$                                        \\
                               &                                      &                                &                                        & $0.3$                                & slow                               & $0.0025$                              & $500$                                        \\
                               &                                      &                                &                                        & $0.3$                                & fast                               & $0.0025$                              & $40000$                                      \\
                               &                                      &                                & \multicolumn{1}{c}{MPL-BD}             & \textit{NA}                          & \textit{NA}                        & \textit{NA}                           & $300000$                                     \\
\hline
$1000$                         & 1\%                                  & $400$                          & MJ-MCMC                                & $0.001$                              & none                               & $1$                                   & $100000$                                     \\
                               &                                      &                                &                                        & $0.01$                               & none                               & $1$                                   & $10000$                                      \\
                               &                                      &                                &                                        & $0.05$                               & none                               & $0.01$                                & $500$                                        \\
                               &                                      &                                &                                        & $0.25$                               & none                               & $0.05$                                & $100$                                        \\
                               &                                      &                                &                                        & $0.75$                               & none                               & $0.05$                                & $200$                                        \\
                               &                                      &                                &                                        & $0.25$                               & slow                               & $0.0025$                              & $3000$                                       \\
                               &                                      &                                &                                        & $0.25$                               & fast                               & $0.0025$                              & $100000$                                     \\
                               &                                      &                                & \multicolumn{1}{c}{MPL-BD}             & \textit{NA}                          & \textit{NA}                        & \textit{NA}                           & $1000000$                                    \\
\hline
$1000$                         & 1\%                                  & $1000$                         & MJ-MCMC                                & $0.001$                              & none                               & $1$                                   & $100000$                                     \\
                               &                                      &                                &                                        & $0.01$                               & none                               & $1$                                   & $10000$                                      \\
                               &                                      &                                &                                        & $0.05$                               & none                               & $0.0025$                              & $5000$                                       \\
                               &                                      &                                &                                        & $0.25$                               & none                               & $0.01$                                & $100$                                        \\
                               &                                      &                                &                                        & $0.75$                               & none                               & $0.05$                                & $200$                                        \\
                               &                                      &                                &                                        & $0.25$                               & slow                               & $0.0025$                              & $3000$                                       \\
                               &                                      &                                &                                        & $0.25$                               & fast                               & $0.0025$                              & $100000$                                     \\
                               &                                      &                                & \multicolumn{1}{c}{MPL-BD}             & \textit{NA}                          & \textit{NA}                        & \textit{NA}                           & $500000$               \\
\hline
\end{tabular}
\caption{\textit{All instances with their corresponding maximum jump size (``max jump") and number of MCMC iterations}}
\label{tab:sim_settings}
\end{table}

\subsection{$p=1000$, $n=400$, $0.2\%$ edge density}

\begin{figure}[H] 
    \centering
    \begin{tabular}[t]{cc}
        \begin{subfigure}[t]{0.5\textwidth}
            \centering            \includegraphics[width=1\linewidth]{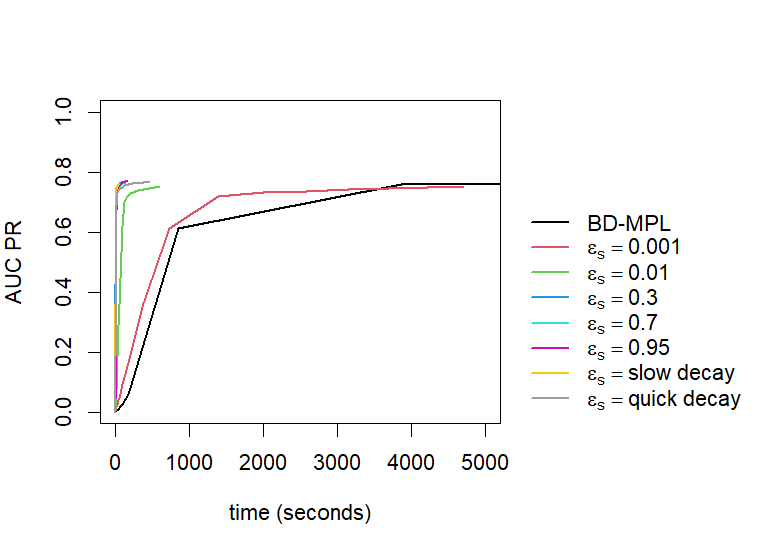}
        \end{subfigure} &
        \begin{subfigure}[t]{0.5\textwidth}
            \centering            \includegraphics[width=1\textwidth]{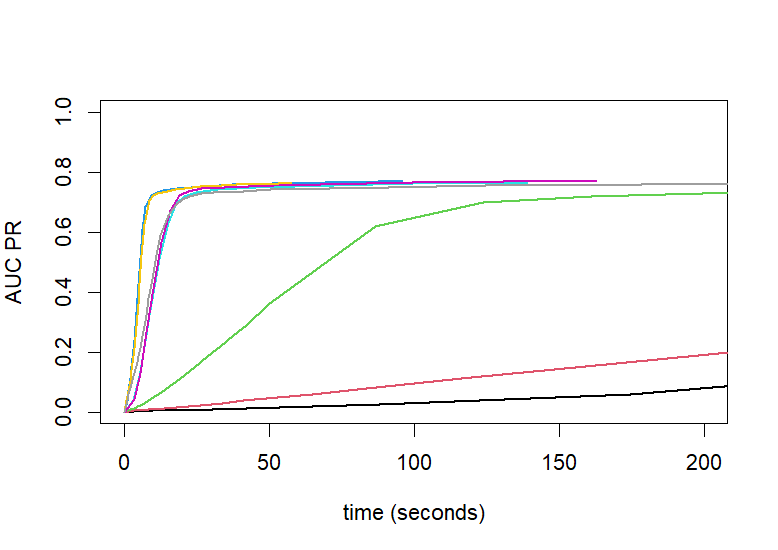}
        \end{subfigure}
    \end{tabular}\\
\caption{\textit{AUC-PR scores over running time for the BD-MPL algorithm and for our approach with $\varepsilon_s = 0.001,0.01,0.3,0.7,0.95$ and for a quickly and slowly decaying $\varepsilon_s$. Zoomed out on the left, zoomed in on the right. Results are the average of eight replications on instances with $p=1000$ variables, $n=400$ observations, and $0.2\%$ edge density. This figure is the same as Figure 2 in the main document.
}}
\end{figure}

\begin{figure}[H] 
    \centering
    \begin{tabular}[t]{cc}
        \begin{subfigure}[t]{0.5\textwidth}
            \centering            \includegraphics[width=1\linewidth]{  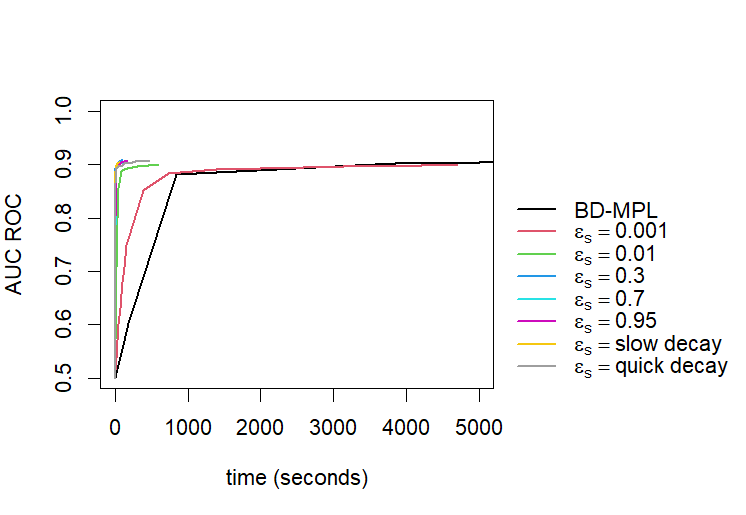}
        \end{subfigure} &
        \begin{subfigure}[t]{0.5\textwidth}
            \centering            \includegraphics[width=1\textwidth]{  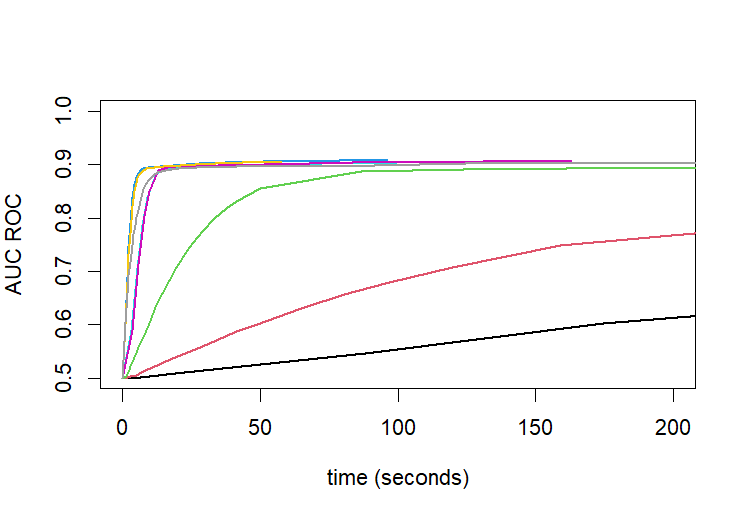}
        \end{subfigure}
    \end{tabular}\\
\caption{\textit{AUC-ROC scores over running time for the BD-MPL algorithm and for our approach with $\varepsilon_s = 0.001,0.01,0.3,0.7,0.95$ and for a quickly and slowly decaying $\varepsilon_s$. Zoomed out on the left, zoomed in on the right. Results are the average of eight replications on instances with $p=1000$ variables, $n=400$ observations, and edge density of $0.2\%$. 
}}
\end{figure}

\begin{figure}[H] 
    \centering
    \begin{tabular}[t]{cc}
        \begin{subfigure}[t]{0.5\textwidth}
            \centering            \includegraphics[width=1\linewidth]{  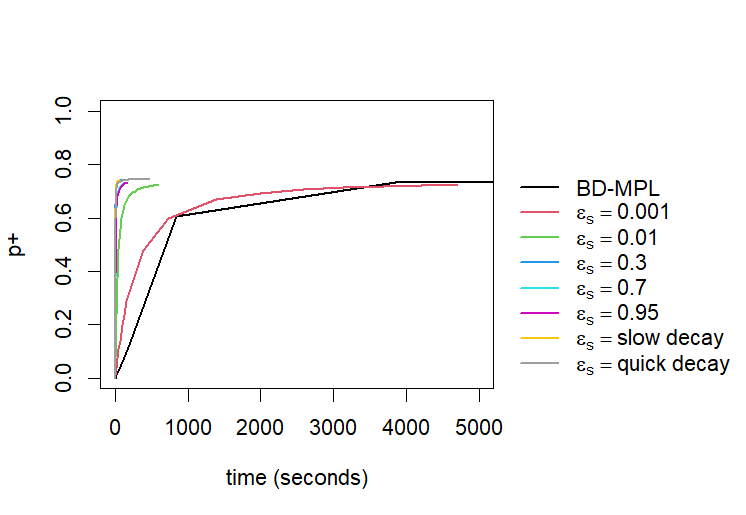}
        \end{subfigure} &
        \begin{subfigure}[t]{0.5\textwidth}
            \centering            \includegraphics[width=1\textwidth]{  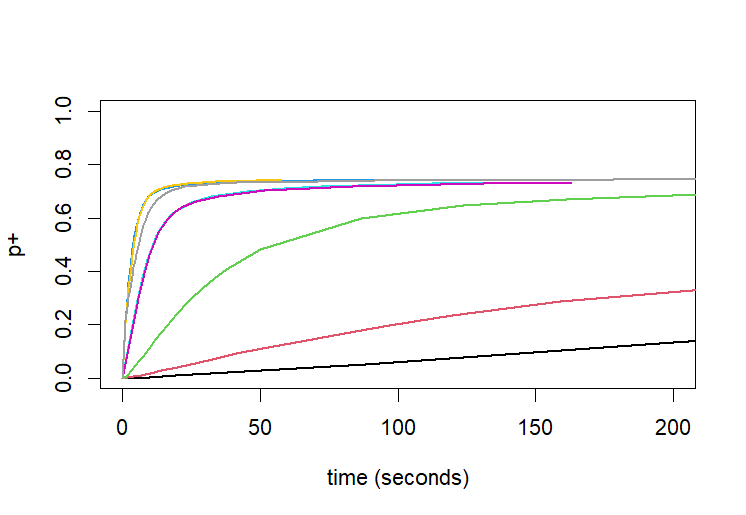}
        \end{subfigure}
    \end{tabular}\\
\caption{\textit{$p^+$ scores over running time for the BD-MPL algorithm and for our approach with $\varepsilon_s = 0.001,0.01,0.3,0.7,0.95$ and for a quickly and slowly decaying $\varepsilon_s$. Zoomed out on the left, zoomed in on the right. Results are the average of eight replications on instances with $p=1000$ variables, $n=400$ observations, and $0.2\%$ edge density.
}}
\end{figure}

\begin{figure}[H] 
    \centering
    \begin{tabular}[t]{cc}
        \begin{subfigure}[t]{0.5\textwidth}
            \centering            \includegraphics[width=1\linewidth]{  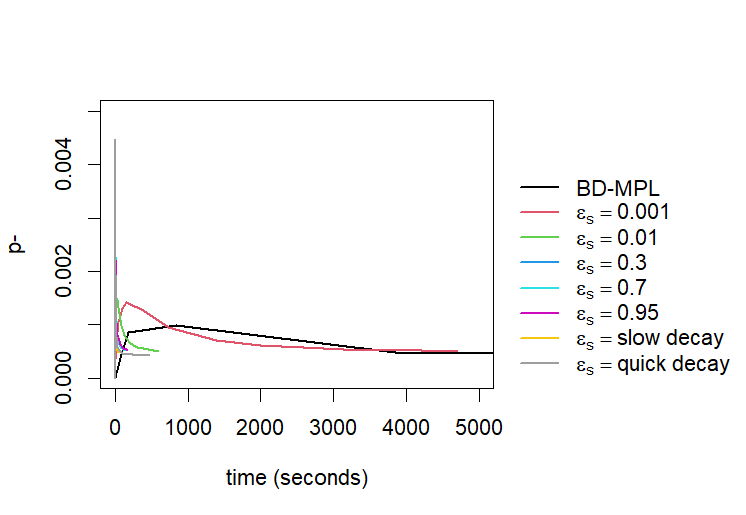}
        \end{subfigure} &
        \begin{subfigure}[t]{0.5\textwidth}
            \centering            \includegraphics[width=1\textwidth]{  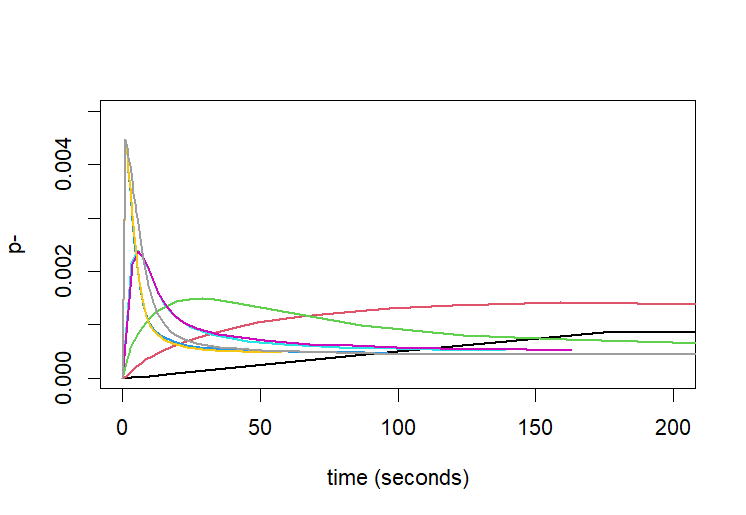}
        \end{subfigure}
    \end{tabular}\\
\caption{\textit{$p^-$ scores over running time for the BD-MPL algorithm and for our approach with $\varepsilon_s = 0.001,0.01,0.3,0.7,0.95$ and for a quickly and slowly decaying $\varepsilon_s$. Zoomed out on the left, zoomed in on the right. Results are the average of eight replications on instances with $p=1000$ variables, $n=400$ observations, and $0.2\%$ edge density.
}}
\end{figure}

\subsection{$p=1000$, $n=1000$, $0.2\%$ edge density}

\begin{figure}[H] 
    \centering
    \begin{tabular}[t]{cc}
        \begin{subfigure}[t]{0.5\textwidth}
            \centering            \includegraphics[width=1\linewidth]{  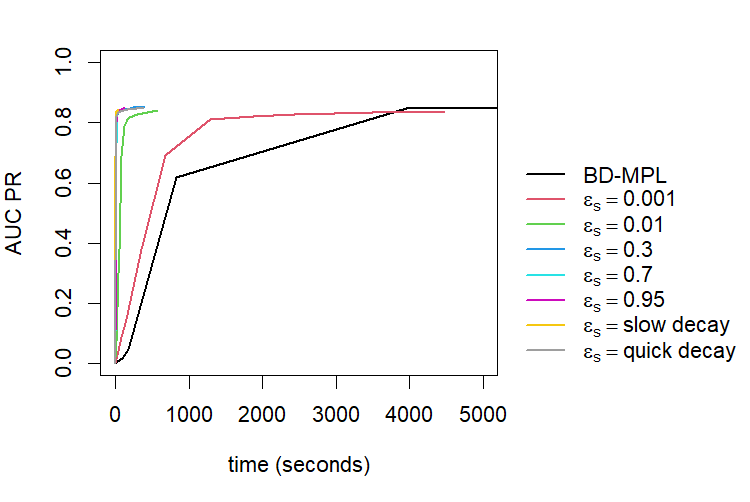}
        \end{subfigure} &
        \begin{subfigure}[t]{0.5\textwidth}
            \centering            \includegraphics[width=1\textwidth]{  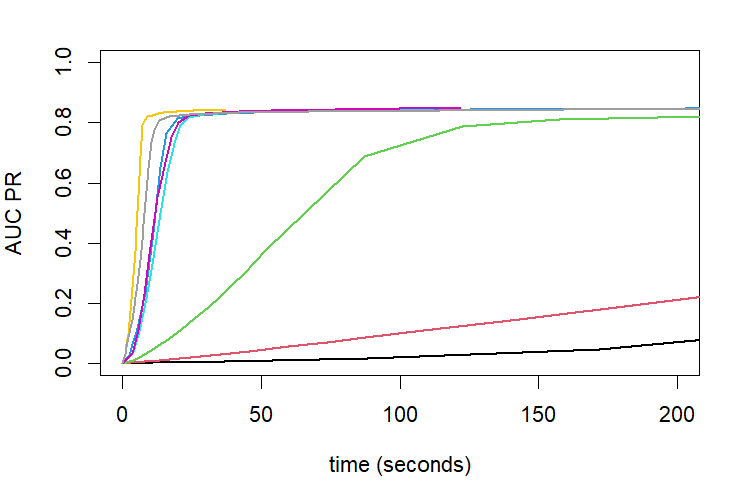}
        \end{subfigure}
    \end{tabular}\\
\caption{\textit{AUC-PR scores over running time for the BD-MPL algorithm and for our approach with $\varepsilon_s = 0.001,0.01,0.3,0.7,0.95$ and for a quickly and slowly decaying $\varepsilon_s$. Zoomed out on the left, zoomed in on the right. Results are the average of eight replications on instances with $p=1000$ variables, $n=1000$ observations, and $0.2\%$ edge density.}}
\end{figure}

\begin{figure}[H] 
    \centering
    \begin{tabular}[t]{cc}
        \begin{subfigure}[t]{0.5\textwidth}
            \centering            \includegraphics[width=1\linewidth]{  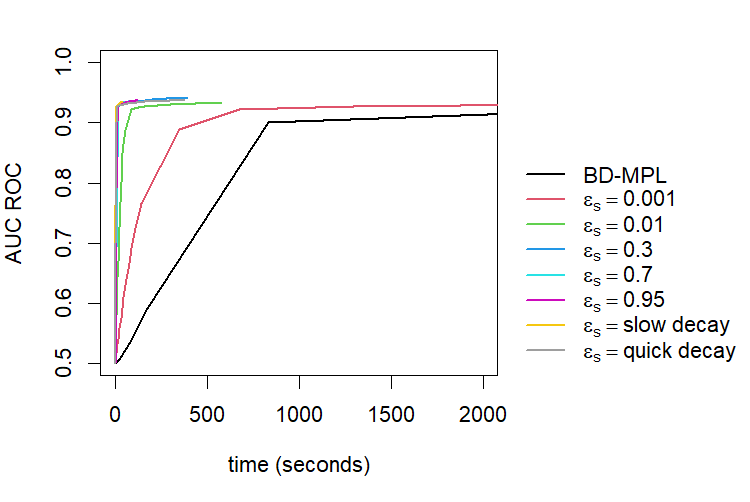}
        \end{subfigure} &
        \begin{subfigure}[t]{0.5\textwidth}
            \centering            \includegraphics[width=1\textwidth]{  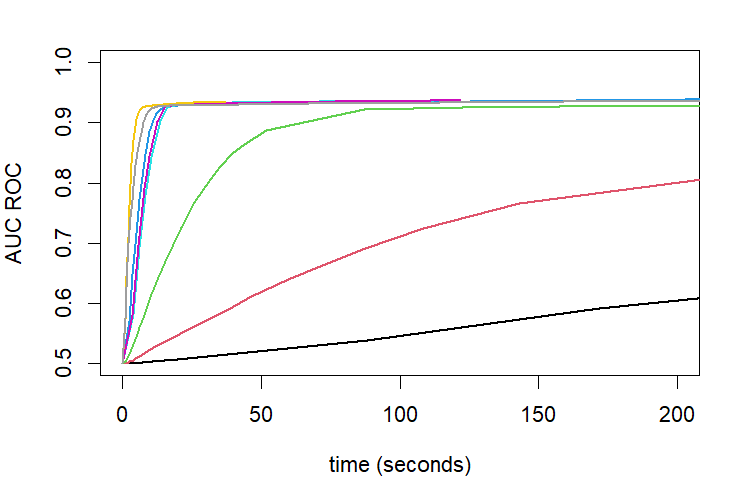}
        \end{subfigure}
    \end{tabular}\\
\caption{\textit{AUC-ROC scores over running time for the BD-MPL algorithm and for our approach with $\varepsilon_s = 0.001,0.01,0.3,0.7,0.95$ and for a quickly and slowly decaying $\varepsilon_s$. Zoomed out on the left, zoomed in on the right. Results are the average of eight replications on instances with $p=1000$ variables, $n=1000$ observations, and edge density of $0.2\%$. 
}}
\end{figure}

\begin{figure}[H] 
    \centering
    \begin{tabular}[t]{cc}
        \begin{subfigure}[t]{0.5\textwidth}
            \centering            \includegraphics[width=1\linewidth]{  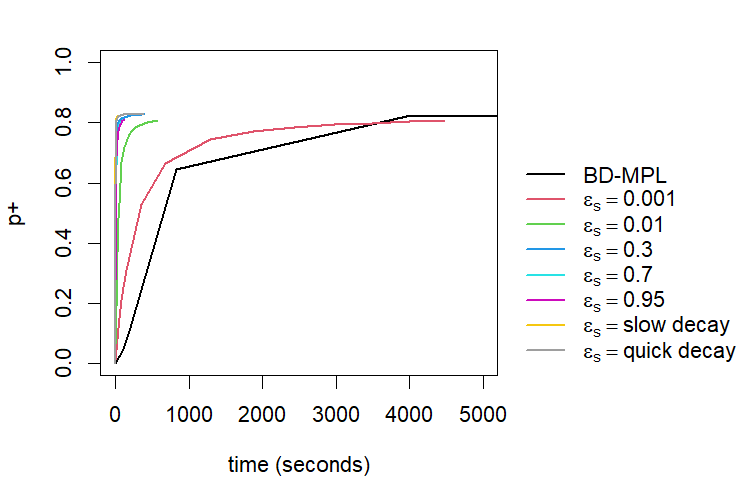}
        \end{subfigure} &
        \begin{subfigure}[t]{0.5\textwidth}
            \centering            \includegraphics[width=1\textwidth]{  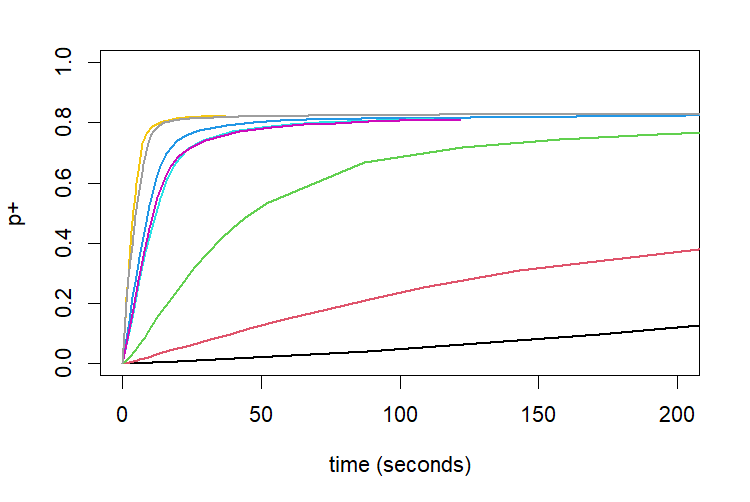}
        \end{subfigure}
    \end{tabular}\\
\caption{\textit{$p^+$ scores over running time for the BD-MPL algorithm and for our approach with $\varepsilon_s = 0.001,0.01,0.3,0.7,0.95$ and for a quickly and slowly decaying $\varepsilon_s$. Zoomed out on the left, zoomed in on the right. Results are the average of eight replications on instances with $p=1000$ variables, $n=1000$ observations, and $0.2\%$ edge density.
}}
\end{figure}

\begin{figure}[H] 
    \centering
    \begin{tabular}[t]{cc}
        \begin{subfigure}[t]{0.5\textwidth}
            \centering            \includegraphics[width=1\linewidth]{  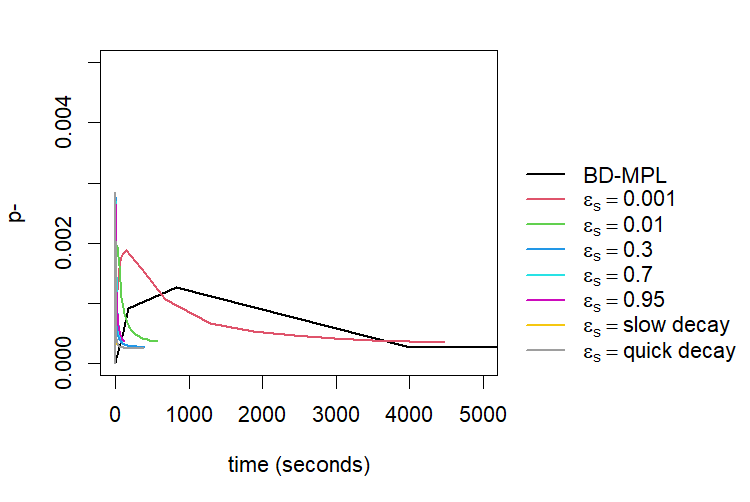}
        \end{subfigure} &
        \begin{subfigure}[t]{0.5\textwidth}
            \centering            \includegraphics[width=1\textwidth]{  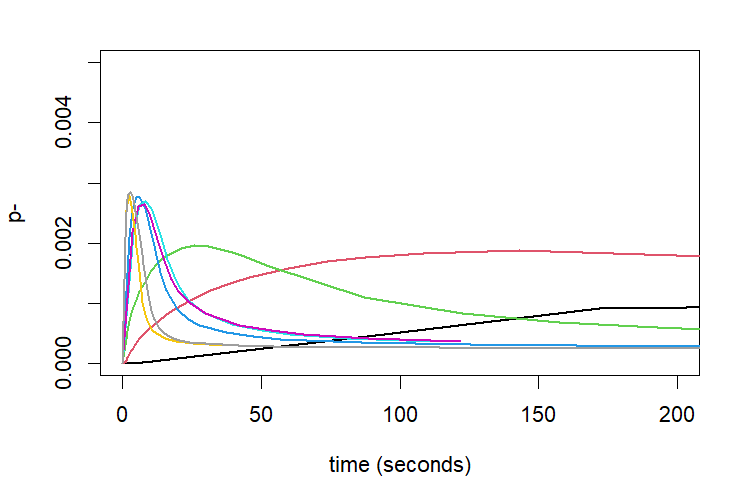}
        \end{subfigure}
    \end{tabular}\\
\caption{\textit{$p^-$ scores over running time for the BD-MPL algorithm and for our approach with $\varepsilon_s = 0.001,0.01,0.3,0.7,0.95$ and for a quickly and slowly decaying $\varepsilon_s$. Zoomed out on the left, zoomed in on the right. Results are the average of eight replications on instances with $p=1000$ variables, $n=1000$ observations, and $0.2\%$ edge density.
}}
\end{figure}

\subsection{$p=1000$, $n=400$, $1\%$ edge density}

\begin{figure}[H] 
    \centering
    \begin{tabular}[t]{cc}
        \begin{subfigure}[t]{0.5\textwidth}
            \centering            \includegraphics[width=1\linewidth]{  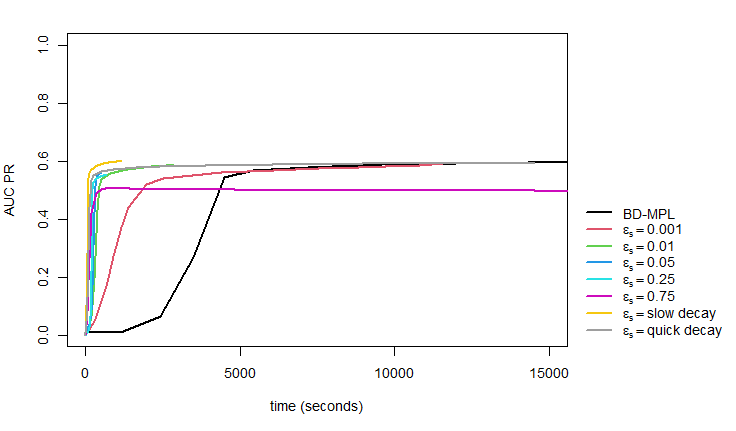}
        \end{subfigure} &
        \begin{subfigure}[t]{0.5\textwidth}
            \centering            \includegraphics[width=1\textwidth]{  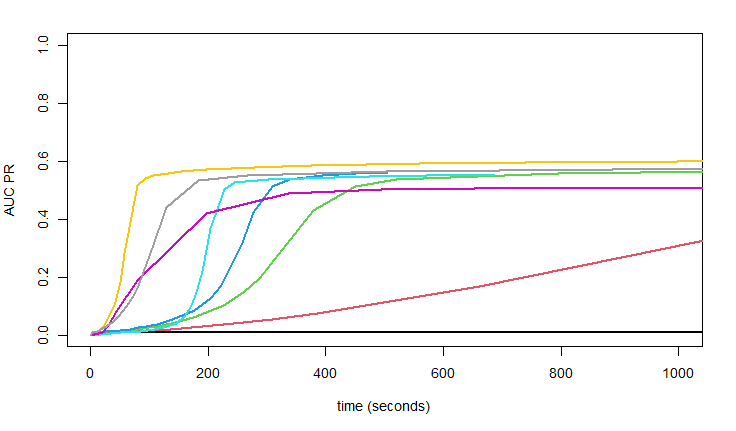}
        \end{subfigure}
    \end{tabular}\\
\caption{\textit{AUC-PR scores over running time for the BD-MPL algorithm and for our approach with $\varepsilon_s = 0.001,0.01,0.05,0.25,0.75$ and for a quickly and slowly decaying $\varepsilon_s$. Zoomed out on the left, zoomed in on the right. Results are the average of eight replications on instances with $p=1000$ variables, $n=400$ observations, and $1\%$ edge density.}}
\end{figure}

\begin{figure}[H] 
    \centering
    \begin{tabular}[t]{cc}
        \begin{subfigure}[t]{0.5\textwidth}
            \centering            \includegraphics[width=1\linewidth]{  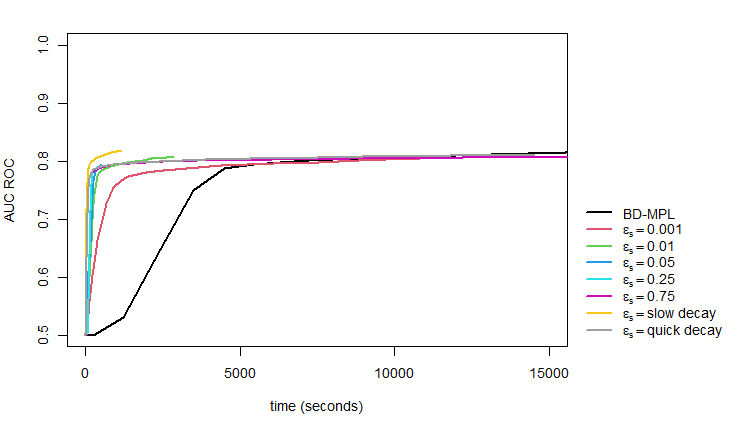}
        \end{subfigure} &
        \begin{subfigure}[t]{0.5\textwidth}
            \centering            \includegraphics[width=1\textwidth]{  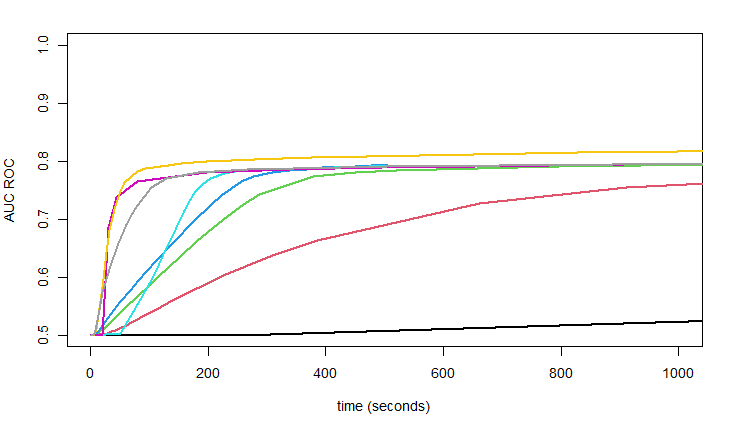}
        \end{subfigure}
    \end{tabular}\\
\caption{\textit{AUC-ROC scores over running time for the BD-MPL algorithm and for our approach with $\varepsilon_s = 0.001,0.01,0.05,0.25,0.75$ and for a quickly and slowly decaying $\varepsilon_s$. Zoomed out on the left, zoomed in on the right. Results are the average of eight replications on instances with $p=1000$ variables, $n=400$ observations, and edge density of $1\%$. 
}}
\end{figure}

\begin{figure}[H] 
    \centering
    \begin{tabular}[t]{cc}
        \begin{subfigure}[t]{0.5\textwidth}
            \centering            \includegraphics[width=1\linewidth]{  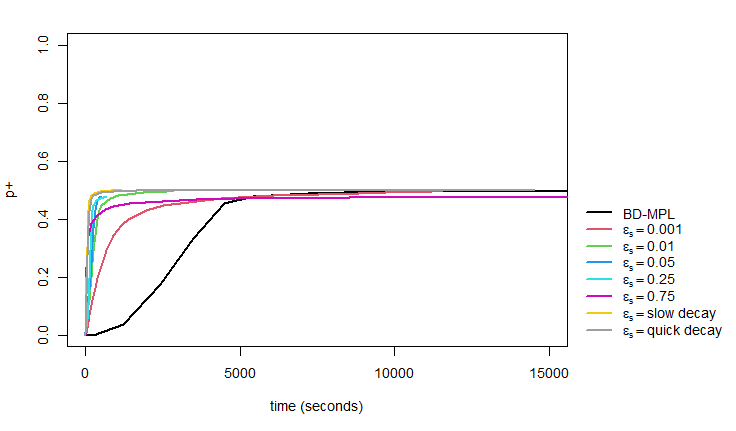}
        \end{subfigure} &
        \begin{subfigure}[t]{0.5\textwidth}
            \centering            \includegraphics[width=1\textwidth]{  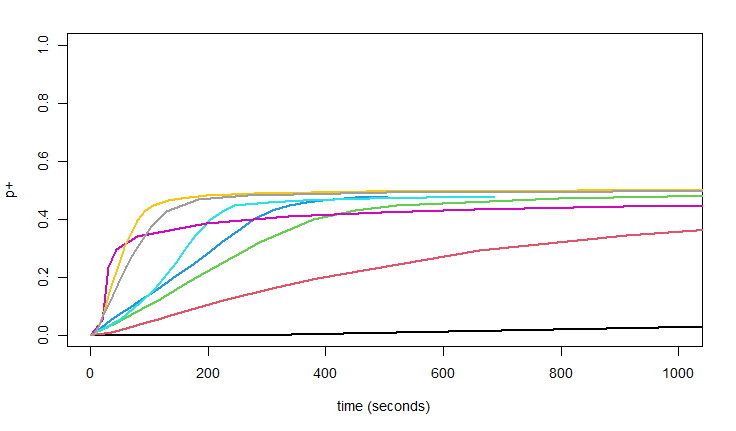}
        \end{subfigure}
    \end{tabular}\\
\caption{\textit{$p^+$ scores over running time for the BD-MPL algorithm and for our approach with $\varepsilon_s = 0.001,0.01,0.05,0.25,0.75$ and for a quickly and slowly decaying $\varepsilon_s$. Zoomed out on the left, zoomed in on the right. Results are the average of eight replications on instances with $p=1000$ variables, $n=400$ observations, and $1\%$ edge density.
}}
\end{figure}

\begin{figure}[H] 
    \centering
    \begin{tabular}[t]{cc}
        \begin{subfigure}[t]{0.5\textwidth}
            \centering            \includegraphics[width=1\linewidth]{  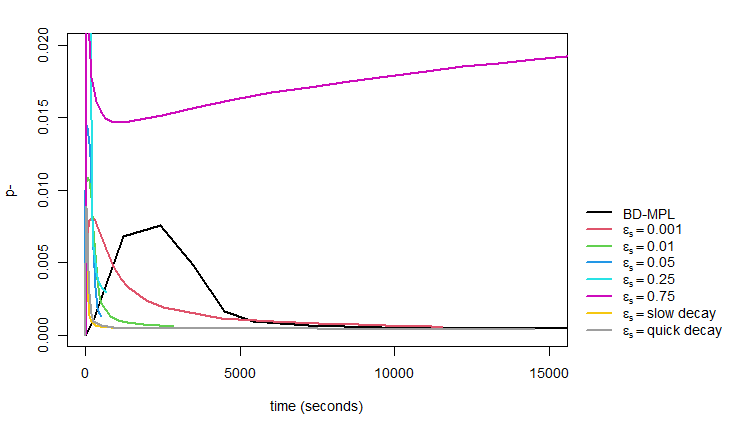}
        \end{subfigure} &
        \begin{subfigure}[t]{0.5\textwidth}
            \centering            \includegraphics[width=1\textwidth]{  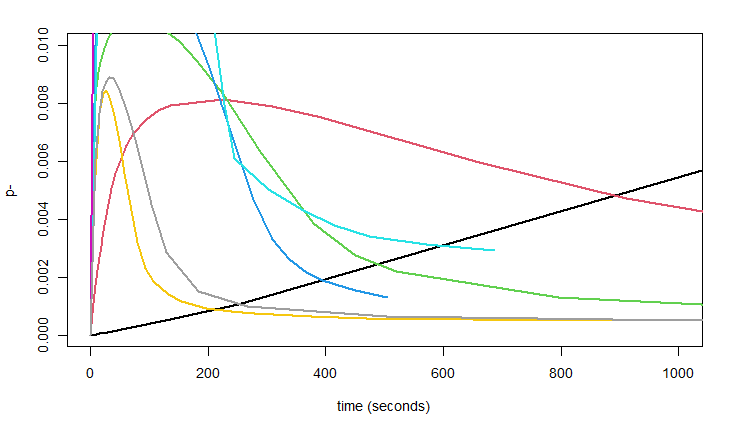}
        \end{subfigure}
    \end{tabular}\\
\caption{\textit{$p^-$ scores over running time for the BD-MPL algorithm and for our approach with $\varepsilon_s = 0.001,0.01,0.05,0.25,0.75$ and for a quickly and slowly decaying $\varepsilon_s$. Zoomed out on the left, zoomed in on the right. Results are the average of eight replications on instances with $p=1000$ variables, $n=400$ observations, and $1\%$ edge density.
}}
\end{figure}

\subsection{$p=1000$, $n=1000$, $1\%$ edge density}

\begin{figure}[H] 
    \centering
    \begin{tabular}[t]{cc}
        \begin{subfigure}[t]{0.5\textwidth}
            \centering            \includegraphics[width=1\linewidth]{  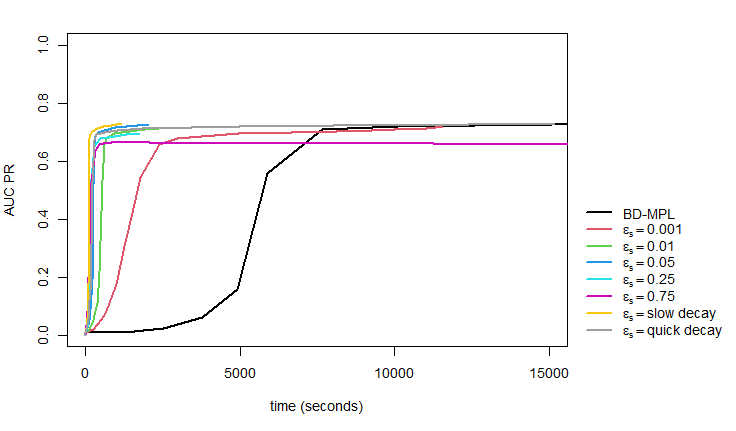}
        \end{subfigure} &
        \begin{subfigure}[t]{0.5\textwidth}
            \centering            \includegraphics[width=1\textwidth]{  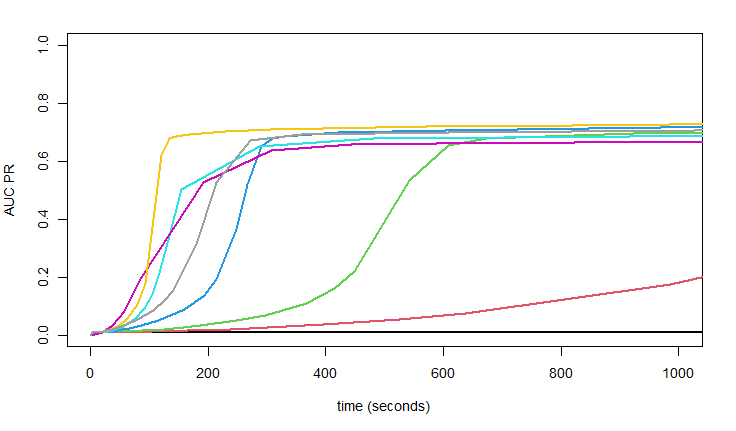}
        \end{subfigure}
    \end{tabular}\\
\caption{\textit{AUC-PR scores over running time for the BD-MPL algorithm and for our approach with $\varepsilon_s = 0.001,0.01,0.05,0.25,0.75$ and for a quickly and slowly decaying $\varepsilon_s$. Zoomed out on the left, zoomed in on the right. Results are the average of eight replications on instances with $p=1000$ variables, $n=1000$ observations, and $1\%$ edge density.}}
\end{figure}

\begin{figure}[H] 
    \centering
    \begin{tabular}[t]{cc}
        \begin{subfigure}[t]{0.5\textwidth}
            \centering            \includegraphics[width=1\linewidth]{  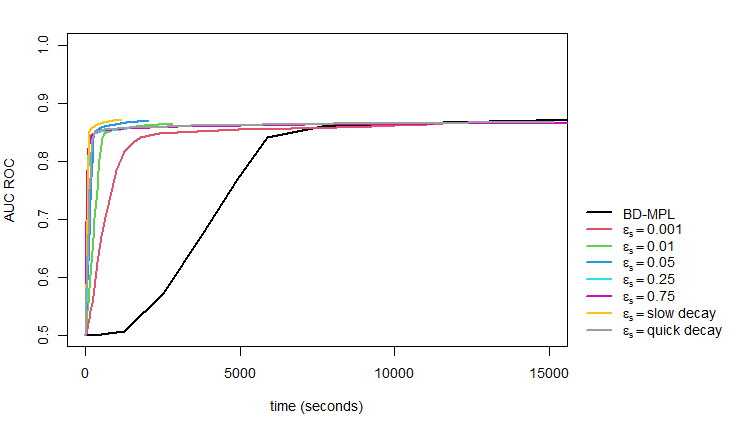}
        \end{subfigure} &
        \begin{subfigure}[t]{0.5\textwidth}
            \centering            \includegraphics[width=1\textwidth]{  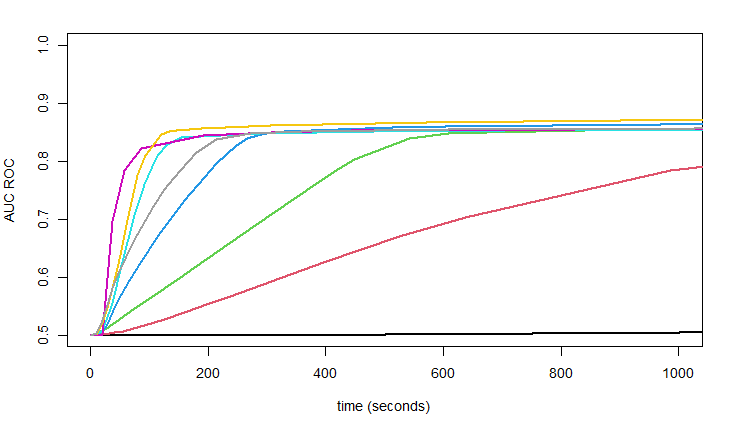}
        \end{subfigure}
    \end{tabular}\\
\caption{\textit{AUC-ROC scores over running time for the BD-MPL algorithm and for our approach with $\varepsilon_s = 0.001,0.01,0.05,0.25,0.75$ and for a quickly and slowly decaying $\varepsilon_s$. Zoomed out on the left, zoomed in on the right. Results are the average of eight replications on instances with $p=1000$ variables, $n=1000$ observations, and edge density of $1\%$. 
}}
\end{figure}

\begin{figure}[H] 
    \centering
    \begin{tabular}[t]{cc}
        \begin{subfigure}[t]{0.5\textwidth}
            \centering            \includegraphics[width=1\linewidth]{  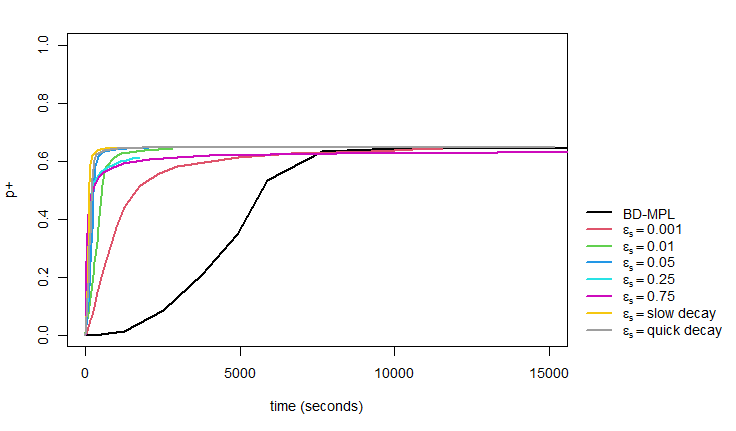}
        \end{subfigure} &
        \begin{subfigure}[t]{0.5\textwidth}
            \centering            \includegraphics[width=1\textwidth]{  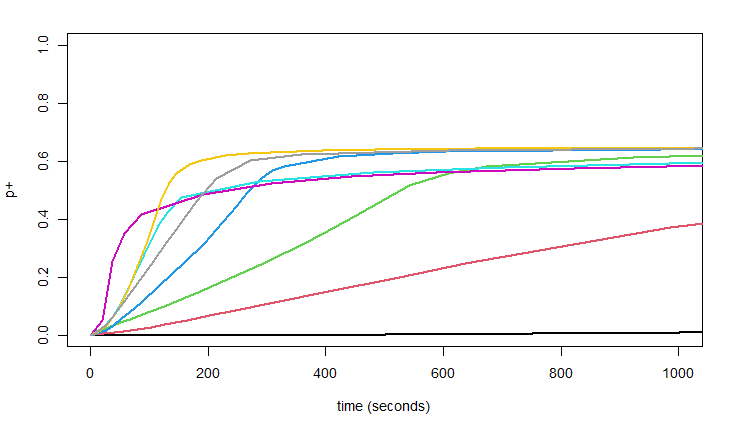}
        \end{subfigure}
    \end{tabular}\\
\caption{\textit{$p^+$ scores over running time for the BD-MPL algorithm and for our approach with $\varepsilon_s = 0.001,0.01,0.05,0.25,0.75$ and for a quickly and slowly decaying $\varepsilon_s$. Zoomed out on the left, zoomed in on the right. Results are the average of eight replications on instances with $p=1000$ variables, $n=1000$ observations, and $1\%$ edge density.
}}
\end{figure}

\begin{figure}[H] 
    \centering
    \begin{tabular}[t]{cc}
        \begin{subfigure}[t]{0.5\textwidth}
            \centering            \includegraphics[width=1\linewidth]{  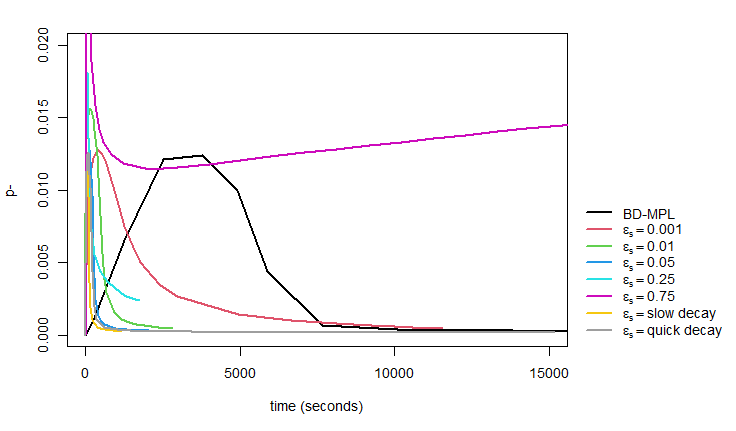}
        \end{subfigure} &
        \begin{subfigure}[t]{0.5\textwidth}
            \centering            \includegraphics[width=1\textwidth]{  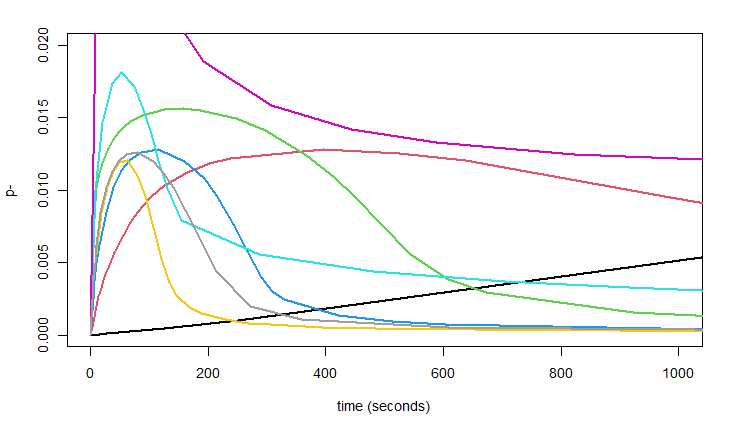}
        \end{subfigure}
    \end{tabular}\\
\caption{\textit{$p^-$ scores over running time for the BD-MPL algorithm and for our approach with $\varepsilon_s = 0.001,0.01,0.05,0.25,0.75$ and for a quickly and slowly decaying $\varepsilon_s$. Zoomed out on the left, zoomed in on the right. Results are the average of eight replications on instances with $p=1000$ variables, $n=1000$ observations, and $1\%$ edge density.
}}
\end{figure}

\bibliographystyle{ba}
\bibliography{sample}